\documentclass[runningheads]{llncs}

\newcommand{\VERSION}{27. 2. 2025}

\usepackage[T1]{fontenc}
\usepackage{graphicx}
\input bksymbol2.sty

\usepackage{mathtools}
\usepackage{xcolor}

\newcommand{\E}{\mathbf{E}}

\newcommand{\KURZ}[2]{#1}

\newcommand{\SmallMath}[1]{\scriptsize #1 \normalsize}

\newcommand{\SAMP}{{\tt SAMP}}
\newcommand{\GAUS}{{\tt GAU}}
\newcommand{\LAP}{{\tt LAP}}
\newcommand{\ADN}{{\tt ADN}}
\newcommand{\ULS}{\UL_{\tt STAT}}
\newcommand{\UL}{{\rm UL}}
\newcommand{\COV}{{\rm Cov}}
\newcommand{\IMG}{{\rm Img}}

\newcommand{\muja}{\mu_{j,\aa}}
\newcommand{\mujb}{\mu_{j,\bb}}
\newcommand{\mujp}{\mu_{j}^{{+}}}
\newcommand{\mujm}{\mu_{j}^{-}}

\renewcommand{\ee}{\varepsilon}
\newcommand{\Fpi}{F^\pi}

\DeclareMathOperator{\supp}{supp}

\begin{document}
\title{Statistical Privacy\thanks{
This research has been conducted within the AnoMed project (https://anomed.de/)
 funded by the BMBF  (German Bundesministerium für Bildung und Forschung)
and the European Union in the NextGenerationEU action.}
}
%
\author{Dennis Breutigam \and R\"udiger Reischuk}
%
\authorrunning{D. Breutigam, R. Reischuk}
%
\institute{University of L{\"u}beck, Ratzeburger Allee 160, 23562 L\"ubeck, Germany\\ 
\email{\{d.breutigam, ruediger.reischuk\}@uni-luebeck.de}\\
\url{www.tcs.uni-luebeck.de}}

\maketitle              
\begin{abstract}
    To analyze the  privacy guarantee of personal data in a data\-base that is
    subject to queries it is necessary to model the prior knowledge of a possible
    attacker. Differential privacy considers a worst-case scenario where he knows
    almost everything, which in many applications is unrealistic and requires a
    large utility loss.
    
    This paper considers a situation called
    \emph{statistical privacy} where an adversary knows the distribution by which
    the database is generated, but no exact data of all (or sufficient many) of
    its entries. We analyze in detail how the entropy of the distribution guarantes privacy for a
    large class of queries called \emph{property queries}.
     Exact formulas are obtained for the privacy parameters.
    We analyze how they depend on the probability that an entry fulfills the property under investigation.
   These formulas turn out to be lengthy, but can be used for tight numerical approximations 
   of the privacy parameters.
   Such estimations are necessary for applying privacy enhancing techniques in practice.
    For this statistical setting we further investigate the effect of
    adding noise or applying subsampling and the privacy utility tradeoff. 
    The dependencies on the parameters are illustrated in detail by a series of plots.
    Finally,  these results are compared to the differential privacy model.
\end{abstract}
\section{Introduction}

In many fields like medicine or social sciences research is not possible without access to personal data.
One solution is to anonymize databases by techniques like microaggregation to achieve
$k$-anony\-mity or variants of it and then make such a modified database publicly available.

Alternatively one could keep the database secret, but allow certain queries about it.
Depending on the type of queries and  prior knowledge
how much information about individual entries can be deduced from
the answers? 

The extreme case where the adversary knows almost everything about the database is
modeled by the \emph{differential privacy} setting \cite{DR14,DMN16}.
This is a pessimistic worst case scenario  rarely occurring in practice where
queries should not be answered precisely because then
the adversary could easily determine the properties of the critical entry.
Instead one has to distort the answer in some way.

What are suitable techniques for the distortion that on the one hand leave much uncertainty 
for an adversary and thus keeps the privacy of individuals, 
but on the other hand still enable researchers
to deduce appropriate results -- the tradeoff privacy versus utility \cite{KL12}?
If in this scenario good privacy can be guaranteed without significant loss of utility
the problem would have been solved. 
However, in many cases the utility loss seems to be too severe.
Hence, more realistic scenarios than this worst case might be better suited to solve
the problem in practice \KURZ{\cite{DP22}}{\cite{DP22}}.
For other subtleties of differential privacy see \cite{M18}.
\KURZ{
Alternatives to this strong privacy notion have been proposed since then.
\cite{DP22} gives a lists of more than 200 such privacy definitions and states their
most important properties and relations as far as known. }{}

The goal of this paper is to mathematically analyze a privacy notion that we consider
most suitable for many practical applications and compare its privacy parameters 
with the one of differential privacy.
Our alternative privacy setting  assumes that an adversary 
does not know all details about the entries. 
His prior information (also called background knowledge) is limited and called
\emph{passive partial knowledge differential privacy} in \cite{DMK20}.
The adversary knows the distribution by which the database has been generated and possibly some
additional information. 
For a rigorous mathematical analysis we will concentrate on the underlying
distribution and do not consider any other unspecified information. 
In  \cite{BBG11} this situation is called \emph{noiseless privacy}.

But we want to consider the option to increase privacy by adding noise or other techniques
and refer to this setting as \emph{statistical privacy}.\footnote{
Note that some authors have used this term for any kind of privacy enhancing modelling.}
It is similar to the notion  \emph{(inference-based) distributional differential
privacy} of \cite{BGKS13}, but instead of a simulator here we compare
conditional distributions directly. {We refer to this setting as
\emph{statistical privacy}. We consider the case that an adversary knows the
distribution of each entry exactly. These distributions can differ.}

Some entries may even be fixed by restricting the support of their distribution to a single value.
Thus, even  background knowledge  that fully specifies certain other entries can be handled.
Differential privacy is the extreme case that all entries are fixed except the critical one.
Then for an adversary there is no uncertainty caused by the entropy of the database distribution.

A suitable setting should be more realistic for many applications.
If the entropy of a given distribution is not sufficient for specific privacy requirements one could
add a privacy mechanism to get better bounds. 
\KURZ{Now t}{T}he hope is that\KURZ{ already}{} for distributions with some entropy good privacy 
properties can be achieved with less distortion by noise\KURZ{ or other techniques}{}.
The goal of this paper to analyze precisely how much privacy amplification can be achieved and how
much utility is lost  for this. 
We consider a generic class of database queries called \emph{counting} or \emph{property queries}.
The databases may have any number of attributes with arbitrary dependencies among them to characterize its entries.
A property query may select  arbitrary combinations of attributes and ask for the number
(percentage) of entries that fulfill this condition -- in the following called 
\emph{positive entries} with respect to this query.

One of our main results is a precise characterization of the privacy loss with respect to the
probability of positive entries. Thus we can give privacy guarantees, for example, for medical data 
when researchers want to investigate rare diseases.

How much the addition of noise generated by classical distributions like Laplace or Gaussian 
can increase the privacy guarantee  
has been investigated in  a series  of papers (for example see \cite{DMN16}).
Subsampling is another technique that has been considered \cite{BBG20,IC21}. 
This paper investigates these mechanisms in the statistical setting.
It turns out that the analysis gets significantly more complicated than in a worst case scenario.
We derive mathematical formulas for the privacy parameters.
In case that they cannot be given in a simple closed form the results of numerical approximations
are presented to compare the different options.
Among others it is shown that the additional entropy in the statistical setting significantly improves
the privacy. 
Furthermore, subsampling compares favorably to artificial noise.
Such explicit bounds have been missing for most privacy notions,
but are essential for application in practice.
In addition the utility loss caused by mechanisms has hardly been investigated.
We evaluate this tradeoff and compare the privacy enhancement of mechanisms
when they generate the same utility loss.
\\

This paper is organized as follows.
The next section introduces the formal setup for querying databases.
In Section~\ref{SectionPriv} we define privacy notions, in particular
introduce the statistical privacy setting.
Then mechanism to amplify privacy are analyzed and we try to obtain
closed formulas for the corresponding parameters and the utility loss.
Since this is not always possible Section~\ref{SectionCompar} and \ref{SectionSubsam} present
the results of numerical approximations of these parameters.
This way we provide a comparison between differential and statistical privacy and
between noise mechanisms and subsampling to better understand the implications in practice.
The paper closes with an outlook for further research on these issues.
\KURZ{}{Due to space limitations most proofs had to be dropped.}

\section{Databases and Queries}

The following setting will be used in this paper to model privacy issues\KURZ{
with respect to querying databases}{}.

\begin{definition}[Databases and Queries]
An \bfinw{entry $I$} of a database is specified by $d$ attributes. 
Formally, it is a vector of the space \wfinw{\begin{math}W \gla W^{(1)} \times \ldots \times W^{(d)}\end{math}},
where \begin{math}
   W^{(i)}
\end{math} are the possible values of the $i$-th attribute. 
Then an entry is given by \wfinw{\begin{math}I = (w_1,\ldots,w_d)\end{math}}
with \begin{math}w_i \in W^{(i)}\end{math}.
There can be any dependencies among the attributes.

A database  \wfinw{$D$} of size $n$  is a sequence of entries \begin{math}
   I_1,\ldots,I_n
\end{math}.
There may be some prior information how a specific database  looks like 
given by a  distribution, resp.~density function \wfinw{\begin{math} \mu=(\mu_1,\ldots,\mu_n)\end{math}} on $W^n$,
where $\mu_j$ is the marginal distribution of the $j$-th entry.
To reduce notation, in this paper we use the same symbols for a distribution and
its density function if it is clear from the context.
By $\muja$ we denote the distribution where \begin{math}
   \mu_j \equiv \aa
\end{math} is 
fixed to a constant value $\aa$ of the support  of $\mu_j$ denoted by \begin{math}
   \supp(\mu_j)
\end{math}.

Let $\caF$ be a set of queries that may be asked for a given database.
Formally, this is described by measurable functions \wfinw{\begin{math}
   F: \; W^n \mapsto A
\end{math}}, 
where $A$ denotes an appropriate  set of possible answers.
\KURZ{\wkast}{}
\end{definition}

A precise  analysis of the information gain of an adversary when querying a database 
looks hopeless if the distribution $\mu$ can be arbitrarily complex.
There may be some dependency between entries of a database, for example if it contains twins
that share many personal attributes.
In many cases it is still realistic to assume that the entries are independent
which will be assumed in the following.
Attributes, however, may have arbitrary dependencies between each other.

If queries are allowed that are specific to certain entries of a database like
``the age of the second entry''\KURZ{ or ``the number of female entries $I_j$
where the position $j$ is  divisible by $3$''}{} it is impossible to guarantee
individual privacy. Thus, we restrict the adversary to queries $F$ where the
order of the entries is irrelevant -- that means \emph{symmetric} functions. In
particular, the weight by which an entry influences the result of $F$ is
identical for all entries. Otherwise\KURZ{ already with}{} simple linear
functions like \KURZ{summing up the}{weighted sums}
values \KURZ{of an attribute over all entries}{}\KURZ{, privacy is completely
lost.}{destroy individual privacy.}
\KURZ{An adversary even not knowing anything about the database  may ask two queries, 
where in the second one the weight of the critical entry is slightly changed,
and from this he can determine the value of that attribute  exactly.}{}
Thus, symmetry is a natural restriction and prevents that precise information about single entries 
can be obtained by such  simple queries.\KURZ{\footnote{
This condition can be relaxed by allowing an entry $I_j$ to be be transformed 
by a fixed function $g_j$ before being used as input.
Thus, we can handle queries $F$ of the form \begin{math}
   F(I_1,\ldots,I_n) \gla h_F( \sum_j
   \; g_j(I_j))
\end{math}, where
the $g_j$ are arbitrary, but fixed for a database, and $h_F$ can be arbitrary, but
has to be symmetric.
}}{}

For symmetric functions the sequence of arguments $I_j$  can be replaced by a multiset
or alternatively by a histogram for all possible values of $W$.
To extract information\KURZ{ from a database}{} one can consider specific \emph{properties} $U$
and ask for the number or percentage of entries that have property $U$\KURZ{ (a
\emph{counting query})}{}.
A property $U$ can be any subset of $W$.
Define \begin{math}
   \mu_j(U) := \sum_{w \in U} \mu_j(w)
\end{math} as the  prior probability that 
the $j$-entry has property $U$.
A property is nontrivial if there is a nonempty set of entries such that
\begin{math}
   \mu_j(U)\not=0
\end{math} or $1$.
In the following we consider only nontrivial property queries. 
For trivial properties one obviously does not have to query the database.

\begin{definition}[Property Query]
A \wfinw{property query} $F$ is described by a subset \wfinw{$W_F \sse W$}.
The correct answer  for $F$ given a database $   D = I_1,\ldots,I_n$  is the value
\(   y_F \gla y_F(D)  \dea | \{j \mena I_j \in W_F\}| / n \ .  \)

 If $D$ is not fixed, but distributed according to a distribution $\mu$ then
 the correct answer is a random variable \wfinw{\begin{math}Y_F^\mu \end{math}}
 with expectation \begin{math}\Erw{Y_F^\mu} =  \sum_j \mu_j(W_F) / n \end{math}.
 In case that all prior probabilities \begin{math}\mu_j(W_F)\end{math} are identical to some value $\pi$
such a property query $F$ will be denoted by \wfinw{$\Fpi$}.

 To protect entries and their attributes one could distort the correct answer by applying
a function $M$ called \bfinw{mechanism} to generate  a distribution around  $y_F(D)$.
 Let \wfinw{\begin{math}Y_{F,M}^\mu \end{math}} denote this random variable.
 \wkast
\end{definition}

\KURZ{Note that $Y_{F,M}^\mu$ has two sources of randomness:
the distribution $\mu$ of $D$ and the deviation generated by $M$.}{}

Property queries are a generic type of queries that can
simulate\KURZ{ almost every}{ most} sensible 
exploration\KURZ{}{s} of databases by symmetric queries.
\KURZ{For example, if one is interested in the average or median value of a certain attribute
listed in a database, by a sequence of property queries with thresholds for this attribute
(the percentage of entries that have at least a certain value for this attribute)
a histogram can be obtained from which the answer to these questions can be approximated
arbitrarily.
This extends to combinations of attributes, too.}{}
With respect to privacy guarantees, properties  with very low or high probability --
\KURZ{for example being very rich or}{} having a rare disease -- could be critical.
This dependency will be part of the following analysis.

A single query $F$  might provide only little information about the entries.
To get more information a database might be queried several times
called \emph{composition} of queries or \emph{interactive protocol} \cite{BBG11,ACG16,M17,S22}.
An important question is the rate by which privacy is reduced in this case.
This paper is restricted to the case of single queries.
Composition in the statistical setting is even more complex to analyze mathematically
and will be subject of further research.

\section{Privacy}
\label{SectionPriv}

To measure the privacy loss of entries contained in a database $D$ after it has been queried
we\KURZ{ want to}{} estimate the information an adversary can deduce from the answers.
This obviously depends on the levels of prior information of the adversary.
\KURZ{He may have no information at all, thus can only choose the  queries to be answered by the curator of the database.}{}

 If an adversary has some background information and already knows certain entries of $D$
let  $D^{n'}$be the database consisting of the remaining $n'$ unknown entries.
For a query $F$ he can consider the restriction of $F$
to $D^{n'}$ and then correct the answer according to the values of the known entries.
Knowing some entries reduces the uncertainty of the adversary.
The privacy analysis can proceed in the same way by taking the marginal distribution
of the database conditioned on the knowledge of the adversary.
In the simplest case we consider only the entropy of the entries that are completely unknown.

The strongest nontrivial adversary one can think of knows all entries  exactly except one $I_j$
and tries to deduce as much information about its attributes as possible.
An easier task is a decision problem, namely to decide given a fixed vector $\aa \in W$ whether the critical entry
has these properties. In other words, whether a known individual with attributes $\aa$ is
contained in the database.
If this is still hard for an adversary the privacy of the critical entry is saved when it equals $\aa$.
This privacy setting been defined as follows.

\begin{definition}[Differential Privacy] \label{DPi}
A query $F$ together with a mechanism $M$ is called
 \bfinw{$(\ee, \delta)$-differential private}   for a set $\caD \sse W^n$ of possible databases 
 if for all \begin{math} (I_1,\ldots,I_n) \end{math},  all \begin{math} \alpha, \beta \in W \end{math}  and \begin{math}j \in [1\ldots n]\end{math}
 such that \begin{math}
   D=(I_1,\ldots,I_j=\aa,\ldots,I_n)
 \end{math} and \begin{math}
   D'=(I_1,\ldots,I_j=\bb,\ldots,I_n)
 \end{math} belong to $\caD$ 
 and all $S \sse   A$, holds
 \begin{displaymath}
   \Prob{Y_{F,M} \in S \mid D}{}   \kla   e^{\ee} \ \Prob{Y_{F,M} \in S \mid D'}  {}+ \delta   \ , 
 \end{displaymath}
in other words the distribution for $D$ and $D'$ are \bfinw{$(\ee,   \delta)$-indistin\-guishable}.\KURZ{\footnote{
There is an alternative definition where the two databases differ by the inclusion of the critical entry.
The privacy results  are similar -- only the security parameters differ slightly.
Comparing the distributions of databases of different sizes complicates the mathematical analysis,
therefore we prefer the first alternative.}
}{}
\end{definition}

If one can guarantee differential privacy with small parameters $\ee,\dd$ against 
such a strong adversary the entries are protected in an optimal way.
However, in general this requires a large distortion by the mechanism $M$ which decreases
the utility of the answer significantly.
In reality, an adversary in most cases has less prior information  about the database.

\begin{definition}[Statistical Privacy]\label{NPi}
A query $F$ together with a mechanism $M$ is called \bfinw{$(\ee,
 \delta)$-statistical private}  for a distribution $\mu$ (or a set of
 distributions $\mu$) if for all  \begin{math}
   j \in [1\ldots n]
 \end{math}, \begin{math}
   S \sse   A
 \end{math}, 
  \begin{math}
   \alpha, \beta \in \supp(\mu_j)
  \end{math}
  holds
\begin{displaymath}
   \Prob{Y_{F,M} \in S }{ \mu_{j ,\aa}}   \kla e^{\ee} \ \Prob{Y_{F,M} \in S}{\mu_{j,\bb}}   + \delta   \ , 
\end{displaymath}
this means the conditional distributions by fixing the $j$-entry to $\aa$, resp.~$\bb$ 
are  $(\ee,   \delta)$-indistinguishable.      
If there is no mechanism involved, that is we consider $Y_F$  distributed
according to   $\mu$, the query $F$ is called  \bfinw{$(\ee,   \delta)$-pure statistical private}.
\end{definition}

Note that our definition of pure statistical privacy is  different from
noiseless privacy\KURZ{ given}{} in \cite{BBG11}
to avoid problems with post-processing as discussed in \cite{M18}.
This issue has\KURZ{ already}{} been observed in \cite{BGKS13} and fixed.
\KURZ{Let us restate two fundamental properties of this privacy notion
called \emph{privacy axioms} in \cite{KL12}.
Using techniques of \cite{DMK20} it can easily been shown.
\begin{lemma}
For every distribution $\mu$ holds:
   \begin{itemize}
      \item \emph{Post-processing:} 
      If a query $F$ and a mechanism $M$ are $(\ee, \delta)$-statistical private
      and $g$ is a random function from the output space $Z$ of queries to a set $Z'$
      then $g \circ (F,M)$ is  $(\ee, \delta)$-statistical private, too.
      \item \emph{Convexity:} 
      For any two queries $F_1$ and $F_2$ with
      mechanisms $M_1$ and $M_2$ that are $(\ee, \delta)$-statistical private
       the query mechanism $(F,M)$ defined by $(F,M) =
      (F_1, M_1)$ with probability $p$ and $(F,M) = (F_2, M_2)$ with probability
      $1-p$ is $ (\ee, \delta)$-statistical private.
   \end{itemize}
\end{lemma}}{Using techniques of \cite{DMK20} it can been shown that
statistical-privacy fulfills post-processing and convexity \cite{KL12}.}
For property queries the privacy notion can be simplified.
For a query $F$ let $\mujp$ (resp.~$\mujm$) be 
the conditional distribution of $\mu$, where
$I_j$ is an entry that belongs  (resp.~does not) to $W_F$.
Now we measure the information gain obtained by the adversary, the privacy loss,
by the difference of $y_{F,M}$ for $\mu_{j,+}$ and $\mu_{j,-}$, and require
\begin{equation}
      \Prob{Y_{F,M} \in S } {\mu_{j,+}} \le
      e^{\ee} \ \Prob{Y_{F,M} \in S } { \mu_{j,-}}  + \delta   \hbox{ \ and} 
      \label{Ce}
\end{equation}
\begin{equation}
     \Prob{Y_{F,M} \in S } {\mu_{j,-}} \le
      e^{\ee} \ \Prob{Y_{F,M} \in S } {\mu_{j,+}}  + \delta   \ . \label{Cz}
\end{equation}

\begin{definition} \label{NPii}
Let  \begin{math}Y_{F,M}: \;  \Omega \ria Z \end{math} be a  random variable for
a property query $F$ and mechanism $M$. $\nu_1,\nu_2$  are distributions on
$\Omega$ and \begin{math} {\nu_i}^{{F,M}}\end{math}  the  distributions,
resp.~density functions of $Y_{F,M}$  with respect  to $\nu_i$. Then for $z \in Z$
the \bfinw{privacy loss random variable} (PLRV)  \cite{SMM19} of \KURZ{$F$ and $M$}{$F,M$}
with resp.~to $\nu_1,\nu_2$ is defined as
\begin{displaymath}
   \caL^{F,M}_{\nu_1,\nu_2}(z) \dea \ln \ \frac{\nu_1^{{F,M}}(z)}{\nu_2^{{F,M}}(z)} \ ,
\end{displaymath}
where 
\begin{math}
   \ln \frac{0}{0} := 0
\end{math} and 
\begin{math}
   \ln \frac{>0}{0} := \infty
\end{math}.
Given $F,M,\nu_1,\nu_2$ we define the \bfinw{privacy curve} by the
following function for $\ee \ge 0$ 
\begin{equation}
   \dd_{F,M,\nu_1,\nu_2}(\ee) \dea \int_Z \mu^{F,M}_{\nu_1}(z) \cdot
   \max ( \; 0, \; 1-\exp(\ee-\caL^{F,M}_{\nu_1,\nu_2}(z))) \mathbf{d} z \ .  \label{PrivCurve}
\end{equation}
If  
\begin{math}
\dd \ge  \underset{j}{\max}   \underset{\aa,\bb \; \in\; \supp(\mu_j)}{\sup} \dd_{F,M,\muja,\mujb}(\ee)
\end{math} 
we call $(\ee,\dd)$ an  \bfinw{achievable privacy pair} for $(F,M)$ and $\mu$.
\end{definition}
In \cite{BBG20}  the privacy curve is called privacy profile.
For property queries the condition on $\dd$ can be simplified to 
\begin{equation}
\dd \ge  \max \{  \dd_{F,M,\mu_+,\mu_-}(\ee), \; \dd_{F,M,\mu_-,\mu_+}(\ee) \} \ .  \label{PrivCurveProp}
\end{equation}
\KURZ{}{Similiar to differential privacy one can show}
\begin{lemma} \label{Le}
For every distribution $\mu$ holds:
If $(\ee,\dd)$ is an achievable pair for $(F,M)$ then $(F,M)$ is $(\ee,\dd)$ statistical private.
\end{lemma}
\KURZ{
   \begin{proof}
       The inequality in (\ref{Ce})  can be rewritten as
       \SmallMath{
       \begin{displaymath}
         \begin{aligned}
            \dd &\ge \max_{S} \left\{\Prob{Y_{F,M} \in S}{\muja} - e^\ee \Prob{Y_{F,M} \in S}{\mujb}  \right\} \\
            &= \max_{S} \left\{\int_S  \muja(z) - e^\ee \  \mujb(z)  \  \mathbf{d}z \right\} 
            = \max_{S} \left\{\int_S  \muja(z)  \ \lk 1- e^\ee \  \frac{ \mujb(z)}{ \muja (z)} \rk  \  \mathbf{d}z \right\}  \ .
         \end{aligned}
       \end{displaymath}}
      The maximizing set $S \sse Z$ it given by 
      \( \dis
         S_{max} := \left\{ z \mid e^\ee  \kla \frac{\muja (z)}{\mujb (z)} \right\}  \ .
   \)
      Therefore the condition is equivalent to
      \SmallMath{
      \begin{displaymath}
            \delta  \ge    \int_{Z} \max \lk 0, \ \muja (z) \ \lk 1- e^\ee \  \frac{ \mujb (z)}{ \muja (z)} \rk  \rk \   \mathbf{d}z = \int_{Z} \max \lk 0, \ \muja (z) \ \lk 1- \frac{e^\ee}{ e^{\ln  \frac{ \muja  (z)}{ \mujb (z)} } } \rk  \rk \  \mathbf{d}z  \ . 
      \end{displaymath}}
   \end{proof}
}{}
\KURZ{Combining Definition~\ref{NPi} and Lemma~\ref{Le} we get}{}
\begin{lemma}\label{lemma:statPrivMaxDelta}
Given $\mu$, a property query $F$ and a mechanism $M$ are 
$(\ee,\dd)$-statistical private iff  for all 
\begin{math}
   j \in [1 \ldots n]
\end{math} and all 
\begin{math}
   \aa,\bb \in \supp(\mu_j)
\end{math} holds
\begin{math}
   \dd \gra  \dd_{F,M,\muja,\mujb}(\ee)\ .
\end{math}
\end{lemma}
The next section derives explicit formulas for $\dd_{F,M,\muja,\mujb}(\ee)$ for property queries $F$
and the mechanisms that have mainly been considered so far.
\KURZ{The pure case without any mechanism can technically be interpreted as a special case of subsampling.
Thus, we start with subsampling and afterwards consider pure statistical privacy.}{}

\section{Mechanisms}

Several mechanisms to guarantee privacy have been analyzed.
One technique is adding external noise to the correct answer of a query.
The amount of noise depends on the diversity of the entries in the database.
If this is large, to achieve differential privacy
the distortion of the answers has to be quite large, too.
On the contrary, in the noiseless privacy setting exact results are returned --
the entropy of the database distribution is used to confuse an adversary.

This section considers property queries and
estimates precisely how much the entropy generated by the distribution 
guarantees statistical privacy and how much this 
can be enlarged by mechanisms based on noise and subsampling.
To keep the mathematical formulas manageable we restrict the
analysis to the case that the probability $\pi$ for fulfilling the property
is identical for all entries.

\subsection{Subsampling}\label{mechanismen:subsampling}

Subsampling means that the answer to a query is derived from a random sample of size $m<<n$ of the entries.
The curator does not compute the exact answer to a query, instead applies
a corresponding function to the subsample drawn.
Let us denote this privacy mechanism  by \SAMP \ --
among others it has been considered in~\cite{WBK19,IC21}.
Some adjustment might be necessary, for example for counting queries
returning the total number since a set of smaller size is considered.
\KURZ{
   For property queries this is not necessary since the ratio is computed.
From an exact answer, however, the size of the subsample might be deduced
since this is a multiple of $1/m$. 
If no additional noise is added this can be prevented if $m$ is chosen as a divisor of $n$.
Our analysis below does not require that $m$ is unknown to an adversary.
Here we consider only sampling  without replacement.
Similar results hold with replacement if $m$ is significantly smaller than $n$
($<\sqrt{n}$).}{}

Given a database $D$ of size $n$ and a property query $F$ we have defined
$y_F = y_F(D)$, resp. $\pi_F$ as the portion of positive elements in $D$.
Let us define  \
\begin{math}   \xi := \pi_F / (1-\pi_F)   \end{math}
as the relation between  positive  and negative elements for a given database $D$.
Furthermore, let $a_{F,m}$ denote the random variable that gives the portion of positive
elements in a random sample $D^m$, which can take values $j/m$ with 
\begin{math}   j \in [0 \ldots m] \end{math}
and is hypergeometric distributed. 
\KURZ{This means
\begin{displaymath}
   \Pro{a_{F,m} = j/m} \gla  \frac{\binom{y_F\; n}{j} \ \binom{(1-y_F) \; n}{m-j}}{\binom{n}{m}}
\end{displaymath}
with expectation  $ y_F$ and variance \begin{math}
   m^{-1} \; y_F \; (1-y_F) \; \frac{n-m}{n-1} \kla 1/4m
\end{math}.
For \begin{math}
   m \le 0.05 \; n
\end{math} the distribution of $a_{F,m}$ is well approximated by the
   binomial distribution for $m$ independent draws  and  success probability $y_F$. Thus,
\(
   \Pro{a_{F,m} = j/m}  \apa    \binom{m}{j} \ y_F^{\ j} \ (1-y_F)^{m-j} \ . 
\)
The expectation of this binomial distribution is given by $y_F$ and its variance by
\begin{math}   m^{-1} \; y_F \; (1-y_F) \end{math}. 
To simplify the calculation one can use this approximation for a fixed database. \\[0.1ex]

Now consider subsampling for a database $D$ 
when it is known that the critical element is positive (resp.~negative).
Due to symmetry one can restrict the analysis to the case that the first element is the critical one.
Let $y_F$ be the ratio of positive elements in the remaining $n-1$ ones.
Then we distinguish whether the first element is drawn or not --
the first case happening with probability 
\begin{math}   \ll = m/n  \end{math}.
This gives
\SmallMath{
\begin{displaymath}
   \begin{aligned}
      \Pro{a_{F,m} = \frac{j}{m} \mena D_1 \in W^F} &\approx  \ll \; \binom{m-1}{j-1} \ y_{F}^{\ j-1}  \;   (1-y_{F})^{m-j}  + (1-\ll) \; \binom{m}{j} \ y_{F}^{\ j}   \;  (1- y_{F})^{m-j} \\ 
      \Pro{a_{F,m} = \frac{j}{m} \mena D_1 \notin W^F} &\approx  \ll \; \binom{m-1}{j} \ y_{F}^{\ j} \ (1-y_{F})^{m-j-1}  + (1-\ll) \; \binom{m}{j} \ y_{F}^{\ j}   \;   (1-y_{F})^{m-j}   
   \end{aligned}
\end{displaymath}}
The quotient
of the right sides can be simplified to
\begin{eqnarray*}
      &&\frac{\Pro{a_{F,m} = j/m  \mena D_1 \in W^F }}{\Pro{a_{F,m} = j/m  \mena
      D_1 \notin W^F}} \\[1ex]
      &&\approx \ \frac{\ll \ \binom{m-1}{j-1} \ y_{F}^{\ j-1}  \;    (1- y_{F})^{m-j}  + (1-\ll) \; \binom{m}{j} \ y_{F}^{\ j}   \;    (1- y_{F})^{m-j} }
          { \ll \ \binom{m-1}{j} \ y_{F}^{\ j} \ (1-y_{F})^{m-j-1}  + (1-\ll) \; \binom{m}{j} \ y_{F}^{\ j}   \;  (1-y_{F})^{m-j}  }   \\[1ex]
      && \gla \frac { \ll \; j \ \;    (1- y_{F})  + (1-\ll) \;m \; y_{F}  \;    (1- y_{F})}
          { \ll \; (m-j) \; \xi \; (1-y_{F} )  + (1-\ll) \;m \; y_{F}  \;    (1- y_{F}) } = \frac { \ll \; j + (1-\ll) \;m \;    y_{F}  }
          { \ll \; (m-j) \;  \xi + (1-\ll) \;m \;    y_{F} }  \  ,
\end{eqnarray*}
which is monoton increasing in $j$. \\[0.1ex]

Now assume that subsampling is performed for a property query $F$ with a random database $D$ 
drawn from a distribution $\mu$ where each entry is positive with probability $\pi_F$. 
Let \begin{math}   \xi = \pi_F / (1-\pi_F)  \end{math} 
and $\mu_{+}$ and $\mu_{-}$ be the marginal
distribution of $\mu$ when  the critical element is positive, resp.~negative.
For these densities we get the same expression as above:}{If subsampling is
performed for a property query $F$ with a random database $D$ drawn from a
distribution $\mu$ where each entry is positive with probability $\pi_F$. Let
$\mu_{+}$ and $\mu_{-}$ be the marginal distribution of $\mu$ when the critical
element is positive, resp.~negative. For these densities we get the following
expressions:}
\SmallMath{
\begin{displaymath}
    \begin{aligned}
        \Prob{a_{F,m} = j/m  }{\mu_{+}} &\gla \ll \ \binom{m-1}{j-1} \ \pi_F^{\ j-1}  \;   (1-\pi_F)^{m-j}  + (1-\ll) \; \binom{m}{j} \ \pi_F^{\ j}   \;  (1- \pi_F)^{m-j} \\[1ex] 
        \Prob{a_{F,m} = j/m  }{\mu_{-}} &\gla \ll \ \binom{m-1}{j} \ \pi_F^{\ j} \ (1-\pi_F)^{m-j-1}  + (1-\ll) \; \binom{m}{j} \ \pi_F^{\ j}   \;   (1-\pi_F)^{m-j}   
   \end{aligned}
\end{displaymath}}
Define 
\SmallMath{
\(  \dis Q_+(j) \dea    \frac{\Prob{a_{F,m} = j/m  }{\mu_{+}}}{\Prob{a_{F,m} = j/m}{ \mu_{-}}}  \ .  \)
}

\begin{lemma} \label{LemSubsam}
When taking a subsample of size $m$\KURZ{ (without replacement)}{} from a database distribution $\mu$,
where each entry independently is positive with probability $\pi_F$, then
the number of positive elements is binomially distributed with parameters $m$ and $\pi_F$.
Therefore, the expectation of  $Y_{F,\SAMP}$ is $\pi_F$ and its variance 
\begin{math}  \pi_F \; (1-\pi_F) / m   \end{math} .\\
Furthermore, for the ratio of the marginal distribution holds
\begin{displaymath}
  Q_+(j) \gla     \frac { \ll \; j + (1-\ll) \;m \;    \pi_{F}  }  { \ll \; (m-j) \;  \xi + (1-\ll) \;m \;    \pi_{F} } \ . \\
\end{displaymath}
\end{lemma}
Since we also have to bound the reciprocal fraction  let us define \begin{math}
   Q_-(j) \dea Q(j)^{-1} 
\end{math}.
Interpolating $Q_+(j) $ to arbitrary rational values $j$ this function takes
value $1$ for \begin{math}   j = \pi_F \; m \end{math} , which is the
expectation of \begin{math}   m \; a_{F,m} \end{math}. To determine the  privacy
curve one has to estimate when the quotients $Q_+(j)$ and $Q_-(j)$ exceed
$e^{\ee}$. For  $Q_+(j)$ consider a fraction 
\begin{math}  j = (1+\cc) \; \pi_F m. \end{math} 
above the expectation where
\begin{math}  0  \le \cc \le 1/\pi_F - 1 \end{math}:
\SmallMath{
   \KURZ{\begin{eqnarray*}
      && Q_+((1+\cc) \; \pi_F \; m)
    \gla  \frac { \ll \; j + (1-\ll) \;m \;    \pi_{F}  } { \ll \; (m-j) \;  \xi + (1-\ll) \;m \;    \pi_{F} }   \\[2ex]
     && \gla  \frac { \ll \; (1+\cc) \; \pi_F \; m+ (1-\ll) \;m \;    \pi_{F}  } 
      { \ll \; (m-(1+\cc) \; \pi_F \; m) \;  \xi + (1-\ll) \;m \;    \pi_{F} }
       \gla   \frac { \ll \; (1+\cc) \; \pi_F \; m+ (1-\ll) \;m \;    \pi_{F}  } 
      { \ll \; (1-(1+\cc) \; \pi_F) \; m \;  \frac{\pi_F}{1-\pi_F}  + (1-\ll) \;m \;    \pi_{F} }   \\[2ex]
      &&\gla  \frac { \ll \; (1+\cc) + (1-\ll) } { \ll \; \frac{1-\pi_F - \cc \; \pi_F}{1- \pi_F}+ (1-\ll) }
     \gla \gla  \frac { \ll \; (1+\cc) + (1-\ll) } { \ll \; (1-\frac{\cc \; \pi_F}{1- \pi_F})+ (1-\ll) }   
      \gla  \frac {    1 + \ll \; \cc}{1 - \ll \; \frac{\cc \; \pi_F}{1-\pi_F}} \ .
\end{eqnarray*}}{\begin{equation*}
Q_+((1+\cc) \; \pi_F \; m)  \gla  \frac {    1 + \ll \; \cc}{1 - \ll \; \frac{\cc \; \pi_F}{1-\pi_F}} \ .
\end{equation*}}
}
This gives the equation 
\begin{math}
   e^\ee \gla Q_+((1+\cc) \; \pi_ F \;  m) \gla  \frac { 1 + \ll \; \cc}{1 - \ll \; \frac{\cc \; \pi_F}{1-\pi_F}} 
\end{math}
that  with respect to  $\cc$ has the solution
\(
   \cc_+^\star \gla  \ll^{-1}  \; \frac{e^\ee-1}{ 1 + \; e^\ee \; \frac{\pi_F}{1-\pi_F}}  \ .
\)
Let  
\begin{math}   \cc_+^{\star\star} := \min\{ \cc_+^\star, \; 1/\pi_F - 1 \} \end{math} 
and 
\begin{math}  j_+^\star \dea (1+\cc_+^{\star\star}) \; \pi_F \; m   \end{math}.
Then for all $   j \le  j_+^\star$ holds: \begin{math}  Q_+(j) \le e^\ee \end{math}.
Now the privacy curve
\begin{math}
   \delta_+(\ee) = \delta_{F, M_{\SAMP} , \mu_+, \mu_-}( \ee )
\end{math}
as defined in (\ref{PrivCurve}) and  (\ref{PrivCurveProp}) simplifies to
\SmallMath{
\begin{align*}
      \dd_+(\ee)
      \KURZ{&= \sum_{j = \left\lceil j_+^\star \right\rceil}^m 
      \Prob{Y_{F,\SAMP} = j/m  }{ \mu_+}  \ \left( 1 - \frac{e^\ee}{Q_+(j)} \right) 
      = \sum_{j = \left\lceil j_+^\star \right\rceil}^m  \left(\Prob{Y_{F,\SAMP} = j/m  }{ \mu_+} 
            - e^{\ee} \; \Prob{Y_{F,\SAMP} = j/m  }{ \mu_-}  \right) \\}{}
     &= \sum_{j = \left\lceil j_+^\star \right\rceil}^m  \pi_F^{j-1} (1-\pi_F)^{m-j-1} \binom{m-1}{j} \lk \ll \lk \frac{j (1-\pi_F)}{m-j} - e^\ee \; \pi_F \rk + (1-\ll) (1-e^\ee) \; \frac{m \; \pi_F (1-\pi_F) }{m-j} \rk 
\end{align*}}%
Similar calculations hold for the negative case 
\begin{math}
\delta_-(\ee) \gla \delta_{F, M_{\SAMP} , \mu_-, \mu_+}( \ee )
\end{math} 
with the monoton decreasing function $Q_{-}(j)$.
\KURZ{Now, for $\cc \le 1$ one has to solve 
\begin{displaymath}
   \begin{aligned}
      Q_-((1-\cc) \; \pi_F \; m) 
&=   \frac { \ll \; (m-j) \;  \xi + (1-\ll) \;m \;    \pi_{F} } { \ll \; j + (1-\ll) \;m \;    \pi_{F}  }   \\[2ex]
&=   \frac { \ll \; (m-(1-\cc) \; \pi_F \; m) \;  \xi + (1-\ll) \;m \;    \pi_{F} } 
      { \ll \; (1-\cc) \; \pi_F \; m+ (1-\ll) \;m \;    \pi_{F}  } 
=   \frac{1 + \ll \; \frac{\cc \; \pi_F}{1-\pi_F}} {    1 - \ll \; \cc}  \ .
   \end{aligned}
\end{displaymath}
The solution is
\(
\dis   \cc_-^\star \gla  \ll^{-1}  \; \frac{1-e^{-\ee}}{ 1 + \; e^{-\ee} \; \frac{\pi_F}{1-\pi_F}}  \ .
\)
For \begin{math}
   \cc_-^{\star\star} := \min\{ \cc_-^\star, \; 1\}
\end{math} and \begin{math}
   j_-^\star \dea (1-\cc_-^{\star\star}) \; \pi_F \; m 
\end{math}
and all \begin{math}
   j \ge  j_-^\star
\end{math} holds: \begin{math}
   Q_-(j) \le e^\ee
\end{math} .
This gives 
\begin{displaymath}
  \dd_-(\ee) \gla
    \sum_{j =0}^{ \lf j_-^\star \rf}
    \Prob{Y_{F,\SAMP} = j/m  }{ \mu_-}  \ \left( 1 - \frac{e^\ee}{Q_-(j)} \right)  \ . 
\end{displaymath}}{}
Summarizing we have shown
\begin{theorem}
For databases of size $n$ and property queries $F$ 
with probability $\pi_F$ for an entry being positive,
subsampling with rate $\ll$   gives the privacy curve $\dd(\ee) = \max \{
\dd_+(\ee), \; \dd_-(\ee) \}$ where
\SmallMath{
\begin{eqnarray*}
      \dd_+(\ee) &\gla& \sum_{j = \left\lceil j_+^\star \right\rceil}^{n \; \ll}  \pi_F^{j-1} (1-\pi_F)^{n \; \ll-j-1} \binom{n \; \ll-1}{j} \\
      &&  \qq  \lk \ll \lk \frac{j (1-\pi_F)}{n \; \ll-j} - e^\ee \pi_F \rk + (1-\ll) (1-e^\ee) \frac{n \; \ll \; \pi_F (1-\pi_F) }{n \; \ll-j} \rk \ , \\[2ex]
      \dd_-(\ee) &\gla& \sum_{j = 0}^{\left\lceil j_-^\star \right\rceil}  \pi_F^{j-1} (1-\pi_F)^{n \; \ll-j-1} \binom{n \; \ll-1}{j} \\
      &&  \qq \lk \ll \lk \pi_F - e^\ee \frac{j (1-\pi_F)}{n \; \ll-j}  \rk + (1-\ll) (1-e^\ee) \frac{n \; \ll \; \pi_F (1-\pi_F) }{n \; \ll-j} \rk   \ , \\
  j_+^\star &\gla&  \min \{ \ll , \; \ll + \frac{e^\ee-1}{ 1 + \; e^\ee \; \frac{\pi_F}{1-\pi_F}} \}  \; n \; \pi_F, \
    j_-^\star \gla \max \{ 0, \; \ll +\frac{e^{-\ee}-1}{ 1 + \; e^{-\ee} \; \frac{\pi_F}{1-\pi_F}} \} \; n \; \pi_F \ . 
\end{eqnarray*}}
\end{theorem}

\KURZ{
\subsection{Pure Statistical Privacy}\label{Pure:chapter}

In the pure statistical  privacy setting there is no mechanism applied to confuse an adversary.
The entropy generated by the distribution of databases of size $n$ and property probability $\pi_F$
is $H(\pi_F) \; n$, resp.~$H(\pi_F) \; (n-1)$ if the critical entry is fixed where
 $H$ denotes the Shannon entropy function.
The distribution of the number of positive entries among the remaining $n-1$ entries 
has variance $\pi_F \; (1-\pi_F) \; (n-1)$.

To derive the $(\ee,\dd)$-curve in this case, for the quotient of the two conditional distributions
$\mu_+$ and $\mu_-$ one can use the same formulas as for subsampling by setting $\ll=1$.
This gives  (assuming $\ee \le \ln 2$)

\begin{theorem}
For databases of size $n$ and property queries $F$ 
with probability $\pi_F$ for an entry being positive,  pure statistical privacy is guaranteed with
$\dd(\ee) = \max \{ \dd_+(\ee), \; \dd_-(\ee) \}$ where
\SmallMath{
\begin{eqnarray*}
  \dd_+(\ee) &\gla& \sum_{j = \left\lceil j_+^\star \right\rceil}^{n}  \pi_F^{j-1} (1-\pi_F)^{n-j-1}
         \binom{n-1}{j}  \lk \frac{j (1-\pi_F)}{n-j} - e^\ee \pi_F \rk  \ , \\[2ex]  
  \dd_-(\ee) &\gla& \sum_{j = 0}^{\left\lceil j_-^\star \right\rceil}  \pi_F^{j-1} (1-\pi_F)^{n-j-1} \binom{n-1}{j}  \lk \pi_F - e^\ee \frac{j (1-\pi_F)}{n-j}  \rk \\
    j_+^\star &\gla& \left( 1 + \frac{e^\ee-1}{ 1 + \; e^\ee \;
    \frac{\pi_F}{1-\pi_F}} \right)   \ n \; \pi_F, \qd 
    j_-^\star \gla \left( 1 + \frac{e^{-\ee}-1}{ 1 + \; e^{-\ee} \;
    \frac{\pi_F}{1-\pi_F}}  \right) \  n \; \pi_F  \ . 
\end{eqnarray*}}
\end{theorem}

Note that the thresholds $j_+^\star$ and $j_-^\star$ are a constant fraction away from the expectation
of the binomial distribution. 
Thus, by Chernoff's bound $\dd(\ee)$ decreases exponentially with respect to $n$.
On the other hand, for small $\ee$ the dependency on $\ee$ is exponentially in $\ee^2$.
In other words, by increasing the database by a factor $\bb$ one can decrease $\ee$ by
a factor $\sqrt{\bb}$ and still achieves about the same $\dd$. 
}{}

\subsection{Noise}\label{AddNoise:chapter}

Adding noise to the result of a query is the standard mechanism to enhance privacy.
It has been shown that Laplace noise can
achieve \emph{pure differential privacy}\KURZ{(that means $\delta=0$)}{}~\cite{DMN16}.

\begin{definition}[Additive Noise Mechanism] \label{AddNoise:definition}
Let $F$ be a property query  and \ADN\ a noise function defined by a
a random variable $\caN$  with mean $0$ that is independent of the database distribution $\mu$.
Then \begin{math}
   Y_{F,\ADN} \dea Y_{F} + \caN
\end{math} is an additive noise mechanism for $F$.
   If \begin{math}
      \caN = Lap(\psi)
   \end{math} is Laplace distributed with scaling factor $\psi$ and variance $2 \psi^2$
   this mechanism will be denoted by  $\LAP_\psi$ and if
   \begin{math}
      \caN = N(0,\psi^2)
   \end{math} is normal distributed with standard deviation $\psi$ by $\GAUS_\psi$.
\end{definition}

The strength of the noise only depends on the sensitivity $s$ of the query, the
maximal difference between the results of two databases that differ in exactly
one entry. For property queries it holds $s=1/n$ for databases of size $n$. In
general, $\ee$-differential privacy can be achieved by Laplace noise
$\LAP_\psi$ with $\psi= s/\ee$, that means we need noise of scale $1/\ee \; n$
for property queries.

As shown the entropy of the database distribution already guarantees a certain
privacy level. How can statistical privacy  be improved by adding noise? This
section establishes   bounds similar to the distributional privacy setting. They
are independent of the probability $\pi$ of a property  -- thus noise may be
helpful for very small $\pi$ where not enough entropy is given.

\begin{lemma} \label{AddNoise:LaplaceBound}
Let $F$ be a property query for a  distribution on databases of size $n$,
where each entry is positive independently with some fixed probability $0<\pi<1$.
Then for  $\psi>0$, $(F,\LAP_\psi)$ is $(1/ \psi n, 0)$--statistical private.
\end{lemma}
\KURZ{\begin{proof}
   For every (measurable) set $S$ holds
   \begin{equation}
      \begin{aligned}
         \Prob{Y_{F,\LAP_\psi} \in S}{\mu_+}
         &= \sum_{j=0}^{n} \Prob{Y_{F,\LAP_\psi} \in S \mena Y_F = \frac{j}{n} }{ \mu_+} \qd \Prob{Y_{F} = \frac{j}{n}}{ \mu_+} \nonumber \\
         &= \sum_{j=1}^{n} \frac{1}{2 \psi} \ \int_S e^{-\frac{\left\vert x - (j/n) \right\vert}{\psi}} \mathbf{d}x \qd \binom{n-1}{j-1} \ \pi_F^{j-1} (1-\pi_F)^{n-j} \nonumber \\
         &= \frac{1}{2 \psi} \ \int_S \ \sum_{j=0}^{n-1} e^{-\frac{\left\vert x - ((j+1)/n) \right\vert}{\psi}} \binom{n-1}{j} \ \pi_F^{j} (1-\pi_F)^{n-j-1} \mathbf{d}x \label{fehlerbestimmung:laplace:1} \\ 
      \end{aligned}
   \end{equation}
   \begin{equation}
      \begin{aligned}
         \Prob{Y_{F,\LAP_\psi} \in S}{\mu_-}
         &= \sum_{j=0}^{n} \Prob{Y_{F,\LAP_\psi} \in S \mena Y_F = \frac{j}{n} }{ \mu_-} \qd \Prob{Y_{F} = \frac{j}{n}}{ \mu_- } \nonumber \\
         &= \frac{1}{2 \psi} \ \int_S \ \sum_{j=0}^{n-1} e^{-\frac{\left\vert x - j/n \right\vert}{\psi}} \binom{n-1}{j} \ \pi_F^{j} (1-\pi_F)^{n-j-1} \mathbf{d}x \label{fehlerbestimmung:laplace:2} \ .
      \end{aligned}
   \end{equation}
   
   Notice that the corresponding summands only differ in the exponential function
   which can be rewritten as 
   \begin{displaymath}
      e^{-\frac{\left\vert x - ((j+1)/n) \right\vert}{\psi}} \gla    e^{-\frac{\left\vert x -  (j/n) \right\vert}{\psi}}
      \qd e^{\frac{1}{\psi}\underbrace{\left(   \left\vert x - j/n \right\vert -
      \left\vert x - (j+1)/n \right\vert \right)}_{\zeta(x, j)}} \ , 
   \end{displaymath}
   where the distinguishing  function $\zeta(x, j)$ is given 
   \begin{displaymath}
      \zeta(x,j) = \begin{cases*}
         -1/n &\text{if $x \leq j/n$}, \\
         +1/n &\text{if  $x \geq (j+1)/n$}, \\
         2 \; (x- \frac{j+1/2}{n} ) &\text{else.}
    \end{cases*}     
   \end{displaymath}
    Since  \begin{math}
      -1/n \le \zeta(x,j) \le1/n
    \end{math} for all $x$ and $j$ we get  
    \begin{displaymath}
         e^{-1/ \psi n}   \    \Prob{Y_{F,\LAP_\psi} \in S}{\mu_-} \kla
         \Prob{Y_{F,\LAP_\psi} \in S}{\mu_+}
         \kla  e^{1/ \psi n}   \    \Prob{Y_{F,\LAP_\psi} \in S}{\mu_-} \ .
   \end{displaymath}
\end{proof}}{}
Note that the bound $e^{1/ \psi n}$ is quite tight.
If  $x \ge 1$ then for all \begin{math}
   0 \leq j \leq n-1
\end{math} holds  \begin{math}
   x \ge (j+1)/n
\end{math}  and thus \begin{math}
   \zeta(x,j) = 1/n
\end{math} .    
Choosing \begin{math}
   S=[1,\infty)
\end{math} gives    
\begin{displaymath}
   \Prob{Y_{F,\LAP_\psi} \in S}{\mu_+}  \gla e^{1/ \psi n}   \    \Prob{Y_{F,\LAP_\psi} \in S}{\mu_-}  \ . 
\end{displaymath} 
 Thus, statistical privacy requires about the same Laplace noise as differential privacy for $\dd=0$.
 \KURZ{But w}{W}hat about \KURZ{$(\ee,\dd)$-privacy with }{}\begin{math}
   \dd>0
 \end{math}?
 In the next section we will show that the uncertainty of an adversary\KURZ{ based on the database distribution}{}
 allows significantly smaller $\dd$ for the same $\ee$.  
 Alternatively, one could add Gaussian noise $\GAUS_{\nu s}$ for arbitrary $\nu$.
\KURZ{In this case  it is not possible to achieve pure differential privacy with $\dd=0$,   
still the normal distribution has other favorable properties. 
In the statistical setting we get the following bound on $\dd$ with resp.~to $\ee$.}{}
\begin{lemma}\label{AddNoise:GaussianBound}
   For databases of size $n$ let $F$ be a property query,
where each entry is positive independently with probability
$0<\pi<1$. Then the query $(F, \GAUS_{\nu s})$ is $(\ee, \delta)$--statistical private
   if $\delta \geq \max \{\dd_{F,M,\mu_+,\mu_-}(\ee),
   \dd_{F,M,\mu_-,\mu_+}(\ee) \}$ where
\SmallMath{
\begin{eqnarray*}
   \dd_{F,M,\mu_+,\mu_-}(\ee) &\approx& \Prob{Y_{F,M} \geq \pi + n \sigma^2 \ee
   +\frac{1/2-\pi}{n}}{\mu_+} - e^\ee \; \Prob{Y_{F,M} \geq \pi + n \sigma^2 \ee
   +\frac{1/2-\pi}{n}}{\mu_-} \\
   \dd_{F,M,\mu_-,\mu_+}(\ee) &\approx& \Prob{Y_{F,M} \leq \pi - n \sigma^2 \ee
   +\frac{1/2-\pi}{n}}{\mu_-} - e^\ee \; \Prob{Y_{F,M} \leq \pi - n \sigma^2 \ee
   +\frac{1/2-\pi}{n}}{\mu_+}
\end{eqnarray*}}
and $\sigma^2 = (\nu s)^2 +  \pi(1-\pi) (n-1)/n^2$.
\end{lemma}
\KURZ{
\begin{proof}
   For larger $n$ the binomial distribution can be approximated by the Gaussian distribution.
   Thus, $Y_F \approx \caN(\pi - \pi/n +
   1/n, \pi (1-\pi) (n-1) /n^2)$ with respect to $\mu_+$ and $Y_F \approx
   \caN(\pi - \pi/n ,\pi (1-\pi) (n-1) /n^2)$ with respect to $\mu_-$.
   Since the convolution of these Gaussians with noise $\GAUS_{\nu s}$ results
   in a Gaussian with variance $\sigma^2 = (\nu s)^2 +  \pi(1-\pi)
   (n-1)/n^2$ it holds
   \begin{eqnarray*}
      Y_{F,M} \mid_{\mu_+} &\approx& \caN(\pi (1-1/n) + 1/n,\sigma^2) \\
      Y_{F,M} \mid_{\mu_-} &\approx& \caN(\pi (1-1/n), \sigma^2).
   \end{eqnarray*}
   This gives
   \begin{eqnarray*}
      \dd_{F,M,\mu_+,\mu_-}(\ee) &\approx \int_A \frac{1}{\sqrt{2 \pi \sigma^2}}
      \max \left\{0, e^{-\frac{(x-\pi(1-1/n)-1/n)^2}{2 \sigma^2}} - e^\ee
      e^{-\frac{(x-\pi(1-1/n))^2}{2 \sigma^2}}  \right\} \mathbf{d}x \\
      &= \int_A \frac{e^{-\frac{(x-\pi(1-1/n))^2}{2 \sigma^2}}}{\sqrt{2 \pi \sigma^2}}
      \max \left\{0, e^{-\frac{-2x n + 2 \pi (n - 1)+ 1}{2 \sigma^2 n^2}} - e^\ee \right\} \mathbf{d}x.
   \end{eqnarray*}
   The second argument of the $\max$ function is monotone increasing and zero
   if
   \begin{equation*}
      \ee = -\frac{-2x n + 2 \pi (n - 1)+ 1}{2 \sigma^2 n^2}.
   \end{equation*}
   This has the solution
   \begin{math}
      x^\star = \pi + n \sigma^2 \ee +(1/2-\pi)/n
   \end{math}
   and reduces the integral to
   \begin{eqnarray*}
      \dd_{F,M,\mu_+,\mu_-}(\ee) &\approx& \int_{x^\star}^\infty \frac{e^{-\frac{(x-\pi(1-1/n)-1/n)^2}{2 \sigma^2}}}{\sqrt{2 \pi \sigma^2}}
      \mathbf{d}x - e^\ee \int_{x^\star}^\infty \frac{e^{-\frac{(x-\pi(1-1/n))^2}{2 \sigma^2}}}{\sqrt{2 \pi \sigma^2}}  \mathbf{d}x \\[2ex]
      &=& \Prob{Y_{F,M} \geq x^\star}{\mu_+} - e^\ee \; \Prob{Y_{F,M} \geq x^\star}{\mu_-}.
   \end{eqnarray*}
   Using similar calculations we get for $x^{\star \star} = \pi - n \sigma^2 \ee
      +(1/2-\pi)/n$
   \begin{equation*}
      \dd_{F,M,\mu_-,\mu_+}(\ee) \approx \Prob{Y_{F,M} \leq x^{\star \star}}{\mu_-} - e^\ee \; \Prob{Y_{F,M} \leq x^{\star \star}}{\mu_+}.
   \end{equation*}
\end{proof}}{}

For a property query $\Fpi$ the expectation of $Y_{F,\GAUS_\psi}$ is $\pi$.
Thus to achieve a small $\dd$, the parameters should be chosen such that 
\begin{math}
   \Prob{Y_{F,\GAUS_\psi}  \ge  \ee n \psi^2 + \frac{1}{2n}}{\mu_+} 
\end{math} is small, that means
the threshold \begin{math}
   \ee n \psi^2 + \frac{1}{2n} 
\end{math} should be significantly larger than the expectation  of 
\begin{math}
   Y_{F,\GAUS_\psi}
\end{math}.
This holds if it is larger than \begin{math}
   \pi + k \psi
\end{math}, some multiple $k$ of the standard deviation.
Note that the bound for $\dd$ involves the error function for which no
simple analytical formula is known.

\subsection{Utility Loss} \label{chapter:utility_loss}  

To estimate the tradeoff between privacy amplification and loss of utility for a
query $F$ and a database $D$, one measures the deviation to the correct answer
$y_F(D)$ by the mean squared error (MSE). In the statistical setting we take the
expectation over all databases with respect to the distribution $\mu$.

\begin{definition}[Utility Loss]
For a query $F$  on databases distributed according to a distribution $\mu$ and
a mechanism $M$ the utility loss  is given by \begin{displaymath}
   \UL(F,M) \dea  \E_{\mu} \left[ \left( Y_{F,M} - Y_{F} \right)^2 \right].
\end{displaymath}
\end{definition}
For subsampling  this  can be estimated as follows.
\begin{lemma} \label{error:DLU:subsampling}
For a property query $\Fpi$ subsampling with selection  probability \begin{math}
   \ll = m/n
\end{math} gives
\ \begin{math}
   \UL(F, \SAMP) \gla \pi (1-\pi)  \  (  \frac{1}{m} - \frac{1}{n} )
\end{math}.
\end{lemma}
\KURZ{\begin{proof}
   The random variables
   $Y_F$ and $Y_{F,M}$ are binomially distributed with expectation $\pi$ and variance 
   \begin{math}
      \pi \; (1-\pi) / n$, resp.~$\pi \; (1-\pi) / m
   \end{math} (Lemma~\ref{LemSubsam}).
   To shorten the formulas we drop writing down the distribution  $\mu$ when it is clear from the context.
   Thus,
   \SmallMath{
   \begin{displaymath}
      \begin{aligned}
         \E_{\mu} \left[ \left( Y_{F,\SAMP} - Y_{F} \right)^2 \right] &= \E \left[ \left( (Y_{F,\SAMP} - {\pi}) - (Y_{F} - {\pi}) \right)^2 \right]  \\[1ex]
         &= \E \left[ \left( Y_{F,\SAMP} - {\pi} \right)^2 \right] + 
              \E \left[ \left( Y_{F} - {\pi} \right)^2 \right] 
             - 2 \; \E \left[ (Y_{F,\SAMP} - {\pi})  (Y_{F} - {\pi})  \right]   \\[1ex]
         &= \frac{{\pi} (1-{\pi})}{m} + \frac{{\pi} (1-{\pi})}{n} -2 \left( {\pi}^2 - {\pi} \; \E \left[ Y_{F,\SAMP} \right] - {\pi} \; \E \left[ Y_{F} \right] + \E \left[ Y_{F} Y_{F,\SAMP} \right]   \right)   \\[1ex]
         &= \frac{{\pi} (1-{\pi})}{m} + \frac{{\pi} (1-{\pi})}{n}  + 2{\pi}^2 
                -2 \; \E \left[ Y_{F} \; Y_{F,\SAMP} \right] \ . \\
      \end{aligned}
   \end{displaymath}}
   The last term can be evaluated  by computing the covariance since
   \begin{displaymath}
      \E \left[ Y_{F} \;  Y_{F,\SAMP} \right] = \E \left[ Y_{F}\right] \  \E \left[  Y_{F,\SAMP} \right] + \COV \left[ Y_{F}, \;  Y_{F,\SAMP} \right]
   \end{displaymath}
    Now using the fact that taken random samples $X$  of a set $D$ and evaluating $F$ on $X$
    has expectation $y_F(D)$ we get
    \SmallMath{
    \begin{displaymath}
      \begin{aligned}
      &\COV \left[ Y_{F}, \;  Y_{F,\SAMP} \right] \gla \E \left[ (Y_{F} - \E[Y_{F}]) \; (Y_{F,\SAMP} - \E[Y_{F,\SAMP}]) \right] \\[1ex]
      &=\int_{W^n} \left(y_F(D) - \E [Y_{F} ]\right)  \left(
        \sum_{\mathrlap{x \in \IMG(\SAMP(D))}} \ 
            \Pro{x = \SAMP(D)} \;   (y_F(x) - \E [Y_{F,\SAMP} ] ) \right)
            \mathbf{d} \mu (D) \\
      &= \int_{W^n} (y_F(D) - \pi) \; (y_F(D) - \pi) \;
         \mathbf{d} \mu (D) = {\rm Var} [Y_F] \gla \frac{{\pi} (1-{\pi})}{n} \ ,  \text{\normalsize which gives}
      \end{aligned}
    \end{displaymath}} 
    \begin{displaymath}
      \begin{aligned}
         \E \left[ Y_{F} \;  Y_{F,\SAMP} \right] =
         {\pi}^2 + \frac{{\pi} (1-{\pi})}{n} \phantom{=} \text{and} \phantom{=}
         \E \left[ \left( Y_{F,\SAMP} - Y_{F} \right)^2 \right]
         = \frac{{\pi} (1-{\pi})}{m} - \frac{{\pi} (1-{\pi})}{n} \ .
      \end{aligned}
    \end{displaymath}
   \end{proof}}{}
\KURZ{ 
From a statistical point of view evaluating a property $F$ on 
databases $D$ could be interpreted differently.
Namely, we want to estimate the expectation $Y_F$ for the given distribution $\mu$.
This introduces a certain MSE.
Now adding a mechanism $M$ like subsampling increases the MSE.
Then the utility loss generated by $M$ should be defined as the increase of the MSE by $M$ 
compared to the MSE without $M$.

\begin{definition}[Statistical Utility Loss]
For a query $F$  on databases distributed according to a distribution $\mu$ and
a mechanism $M$ the statistical utility loss  is defined as 
$  \ULS(F,M) \dea  {\rm MSE}(Y_{F,M}) -  {\rm MSE}(Y_{F}) \ . $
\end{definition}
This value can easily be computed.

\begin{lemma} \label{error:SLU:subsampling}
Let $\Fpi$ be a property query   and $\SAMP$ be  the subsampling mechanism $\SAMP$ with selection
   probability \begin{math}
      \lambda = m/n
   \end{math}. 
   Then the statistical utility loss equals
   \begin{displaymath}
      \ULS(F, \SAMP) \gla \frac{\pi (1-\pi)}{m} - \frac{\pi (1-\pi)}{n}. 
   \end{displaymath}
\end{lemma}

 \begin{proof}
As the subsampling mechanism is independend of $\mu$ Lemma~\ref{LemSubsam} gives
\begin{displaymath}
   \ULS(F, \SAMP) \gla \text{MSE}(Y_{F,M}) - \text{MSE}(Y_{F}) \gla  \frac{\pi(1-\pi )}{m} - \frac{\pi  (1-\pi )}{n}.
\end{displaymath}
\end{proof}

It may look surprising that for subsampling of $\mu$-distributed databases 
for property queries the utility loss is identical to the statistical utility loss.
For additive noise with mean $0$ this is quite easy to see.

\begin{lemma} \label{error:DLU:add_noise}
 Let $\Fpi$ be a property query 
and $\ADN$ an additive noise mechanism defined by the random variable
$\caN_{\ADN}$. Then it holds \begin{math}
   \UL(F, \ADN) \gla  \Va (\caN_{\ADN}) \gla \ULS(F, \ADN)
\end{math}.
\end{lemma}

\begin{proof}
Since  $\caN_{\ADN}$ has mean  $0$ and is independent of  the database distribution  the expected difference is 
\begin{displaymath}
      \E \left[ \left( Y_{F,\ADN} - Y_{F} \right)^2 \right] \gla 
      \E \left[ \left( Y_{F} + \caN_{\ADN} - Y_{F} \right)^2  \right] \gla \E \left[  \caN_{\ADN}^2  \right] \ .
\end{displaymath}
Similarly, since the noise has mean $0$, for  an unbiased estimator $Y_{F,\ADN}$  for $\pi$ 
 we only need to calculate its variance which equals the mean square error.
 \begin{displaymath}
   \begin{aligned}
      \text{MSE}(Y_{F,\ADN}) &\gla \E \left[ \left( Y_{F} - \pi \right)^2 \right]
      +2\;  \E \left[\caN_{\ADN} \left( Y_{F} - \pi \right) \right] + \E \left[ N_{\ADN}^2 \right] \\
      &\gla \text{MSE}(Y_F) + \text{Var}(\caN_{\ADN}) \ . 
   \end{aligned}
 \end{displaymath}

\end{proof}
}{\begin{lemma} \label{error:DLU:add_noise}
   Let $\Fpi$ be a property query 
  and $\ADN$ an additive noise mechanism defined by the random variable
  $\caN_{\ADN}$. Then it holds \begin{math}
     \UL(F, \ADN) \gla  \Va (\caN_{\ADN})
  \end{math}.
  \end{lemma}}

In the following let us compare the different privacy mechanisms for a fixed privacy loss.
As shown\KURZ{ above}{}, subsampling generates \begin{math}
   \UL(F, \SAMP) \gla \pi (1-\pi)  \  (  \frac{1}{m} - \frac{1}{n} )
\end{math}
for property queries\KURZ{ $\Fpi$}{}, whereas for additive noise the utility loss is given by the variance.
For Laplace noise this gives\KURZ{ the equation}{} 
\begin{math}
   2 \psi^2 \gla \pi (1-\pi)  \  (  \frac{1}{m} - \frac{1}{n} ) \ .
\end{math}
Since $(F,\LAP_\psi)$ is $(1/ \psi n, 0)$--statistical private (Lemma~\ref{AddNoise:LaplaceBound})
the Laplace mechanism achieves 
\begin{displaymath}
   \ee \gla \sqrt{\frac{2}{\pi (1-\pi)}} \  \sqrt{\frac{\ll}{n-m}} \ . 
\end{displaymath}
Subsampling can achieve the same $\ee$, but with a positive $\dd$.
A formula for this $\dd$ seems to be quite complicated.
\KURZ{For \begin{math}
   \ee < 1.79
\end{math} we could deduce the expression
\begin{displaymath}
   \begin{aligned}
     & \sum_{j = \left\lceil \min\left\{ \frac{m \pi \sqrt{\ll^{-1} 2
      \xi^{-1}}}{\sqrt{n-m} + \sqrt{2 \xi \ll} + 2 \ll / \sqrt{n-m}}, \; m - \pi \right\} \right\rceil}^m \hspace*{-1.5cm}  
      \lk \pi^{j} (1-\pi)^{m-j} \binom{m}{j}  (1 + \ll ( \frac{j}{m} \pi^{-1} - 1) ) \rk \\
     & \qq\qq  \lk \frac{ \ll \left( j + (m-j) \xi + \frac{2 \ll}{1-\pi^2}
      \frac{m-j}{n-m} + \sqrt{\ll \xi} \left( \frac{m-j}{\sqrt{n-m}}
      \frac{\sqrt{2}}{1-\pi} - 2m \right) \right) }{ \ll j - (1- \ll) m \pi } \rk \ ,
   \end{aligned}
\end{displaymath}
for $\dd$, but this does not look helpful.}{}
In the next section we present computations for concrete parameter settings and show
plots of the results.
Now, let us make the same comparison between subsampling and Gaussian noise.

\begin{lemma}
Let $\Fpi$ be a property query 
 and $\SAMP$  a subsampling mechanism with sample size $\ll n$. 
 $\GAUS$ denotes an Gaussian noise mechanism such that  the same utility loss occurs as for $\SAMP$. 
 Then $(F, \GAUS)$ is  $(\ee, \delta)$--statistical private
 if $\delta \geq \max \{\dd_{F,M,\mu_+,\mu_-}(\ee),
   \dd_{F,M,\mu_-,\mu_+}(\ee) \}$ where with $\psi = \pi(1-\pi)(1/\ll -1/n)$
\SmallMath{
\begin{eqnarray*}
   \dd_{F,M,\mu_+,\mu_-}(\ee) &\approx& \Prob{Y_{F,M} \geq \pi + \ee \psi   
   +\frac{1/2-\pi}{n}}{\mu_+} - e^\ee \; \Prob{Y_{F,M} \geq \pi +\ee \psi 
   +\frac{1/2-\pi}{n}}{\mu_-} \\
   \dd_{F,M,\mu_-,\mu_+}(\ee) &\approx& \Prob{Y_{F,M} \leq \pi -\ee \psi 
   +\frac{1/2-\pi}{n}}{\mu_-} - e^\ee \; \Prob{Y_{F,M} \leq \pi - \ee \psi 
   +\frac{1/2-\pi}{n}}{\mu_+}.
\end{eqnarray*}}

\end{lemma}
\KURZ{
\begin{proof}
Follows from Lemma \ref{AddNoise:GaussianBound},  \ref{error:DLU:subsampling} and \ref{error:DLU:add_noise}.
\end{proof}
Again it does not seem possible to describe $\dd$ by a simple explicit expression.}{}

\section{Differential Privacy versus Statistical Privacy}
\label{SectionCompar}

The previous section shows that the relation between the relevant parameters for privacy estimations
often cannot be expressed by simple formulas.
Therefore we have done computations using Mathematica to illustrate these dependencies for the
mechanisms discussed above. In particular, this should provide a practically useful comparison of differential 
and statistical privacy.

There does not seem to be a consensus on the range of $\ee$ needed to provide acceptable
protection in the differential privacy setting. 
Many authors consider $\ee$-values below $1$,
but some real-world data providers use values $\ee>5$ and claim 
still to achieve sufficent  privacy, for example \cite{D21,JD22}. 
For a worst case scenario where the adversary knows almost everything this
definitely cannot be true.
\KURZ{A calculation of  the success probability of an adversary for
different $\ee$-values shows that $\ee>1$ is useless.}{}

Differential privacy assuming that an adversary knows all entries except the critical one 
 is prone to exaggerate the real privacy loss in practical applications.
 However, mathematically it is relatively simple to estimate the privacy parameters when
 Laplace noise is used.
 For statistical privacy the situation is far more complex since one has to compute the
 convolution of the database distribution -- which could be considered as a kind of
 \emph{internal noise} -- and the noise added by the curator -- \emph{external noise} 
 as we seen in the previous section.

To understand the impact of internal noise in the statistical privacy setting, 
we have  compared differential and statistical privacy for property queries 
when additive noise\KURZ{, either Gaussian or Laplace,}{} is added. 
The computations were done for databases of size \begin{math}
   n=1000
\end{math} 
and \KURZ{privacy parameter }{}\begin{math}
   \varepsilon = 0.01
\end{math}.
For property queries the sensitivity equals \begin{math}
   s=1/n
\end{math}. 
In addition, we have evaluated the bounds in case of small databases with \begin{math}
   n=100
\end{math}\KURZ{, too}{}.

For different strength \begin{math}
   \sigma = \nu \; s
\end{math} of external noise the $\dd$ parameters achieved in
both cases are plotted, where $\sigma$ denotes the standard deviation for Gaussian noise
and the scaling factor in case of Laplace noise.
For moderate values \begin{math}
   \nu_1=1
\end{math} and \begin{math}
   \nu_2=3
\end{math} the utility loss seems to be acceptable.

In the statistical setting the results depend on the property probability $\pi$.
If $\pi$ is very close to $0$ or to $1$ the situation is similar to the worst case scenario
in differential privacy. Since property queries (and the entropy function) are symmetric with respect to 
$ \pi=1/2$.
It suffices to restrict $\pi$ to the interval$ [0,1/2]$.
\begin{figure}[ht!]
   \centering
   \begin{minipage}{.45\textwidth}
        \centering
        \includegraphics[width=.95\linewidth]{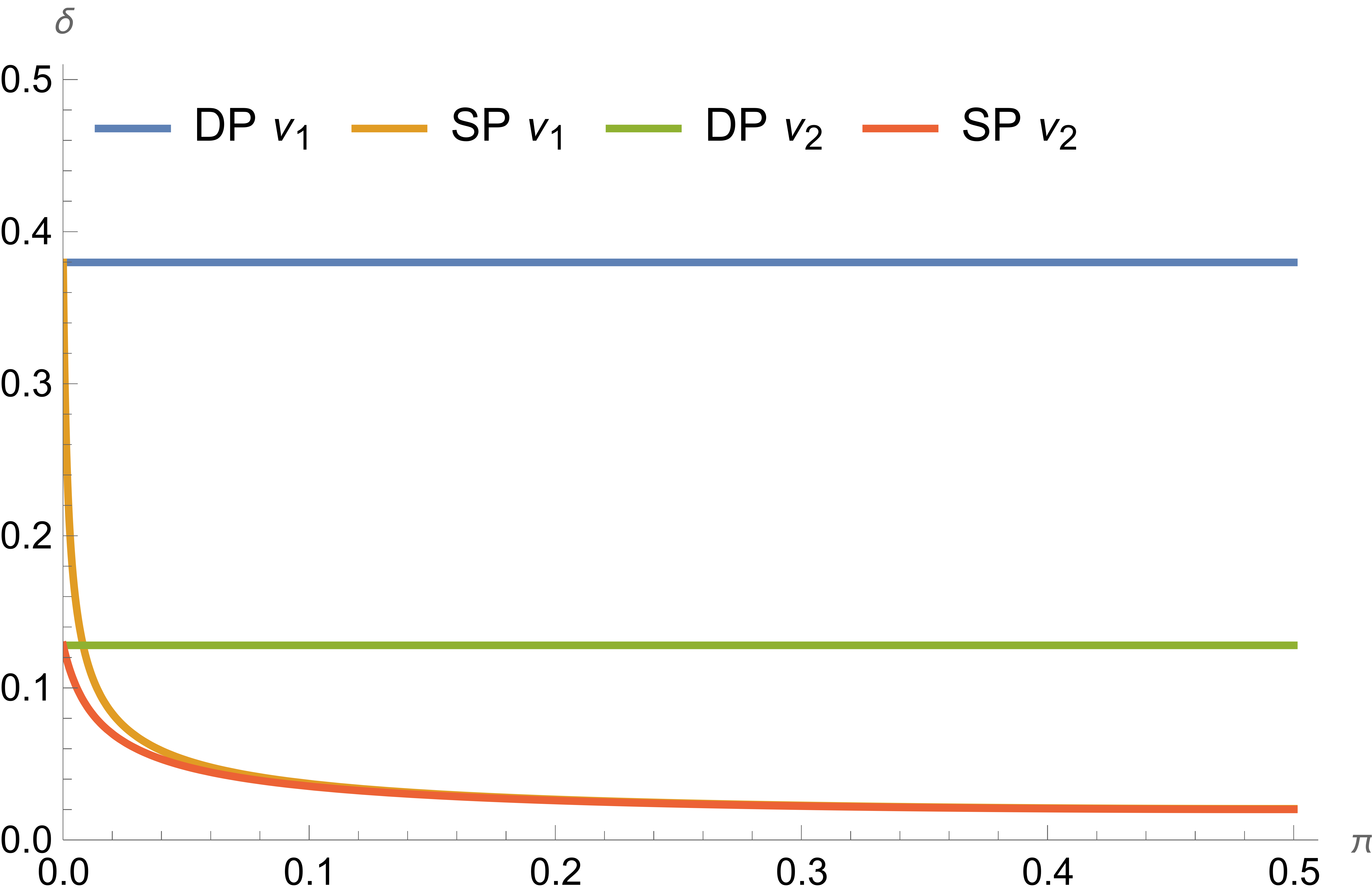}
        \caption{$\delta$ parameter with respect to $\pi$ for noise levels $\nu_1=1$ and  $\nu_2=3$.}
        \label{figure:gaus_dp_vs_sa}
   \end{minipage}
   \hspace*{.05\textwidth}
   \begin{minipage}{.45\textwidth}
        \centering
        \includegraphics[width=.95\linewidth]{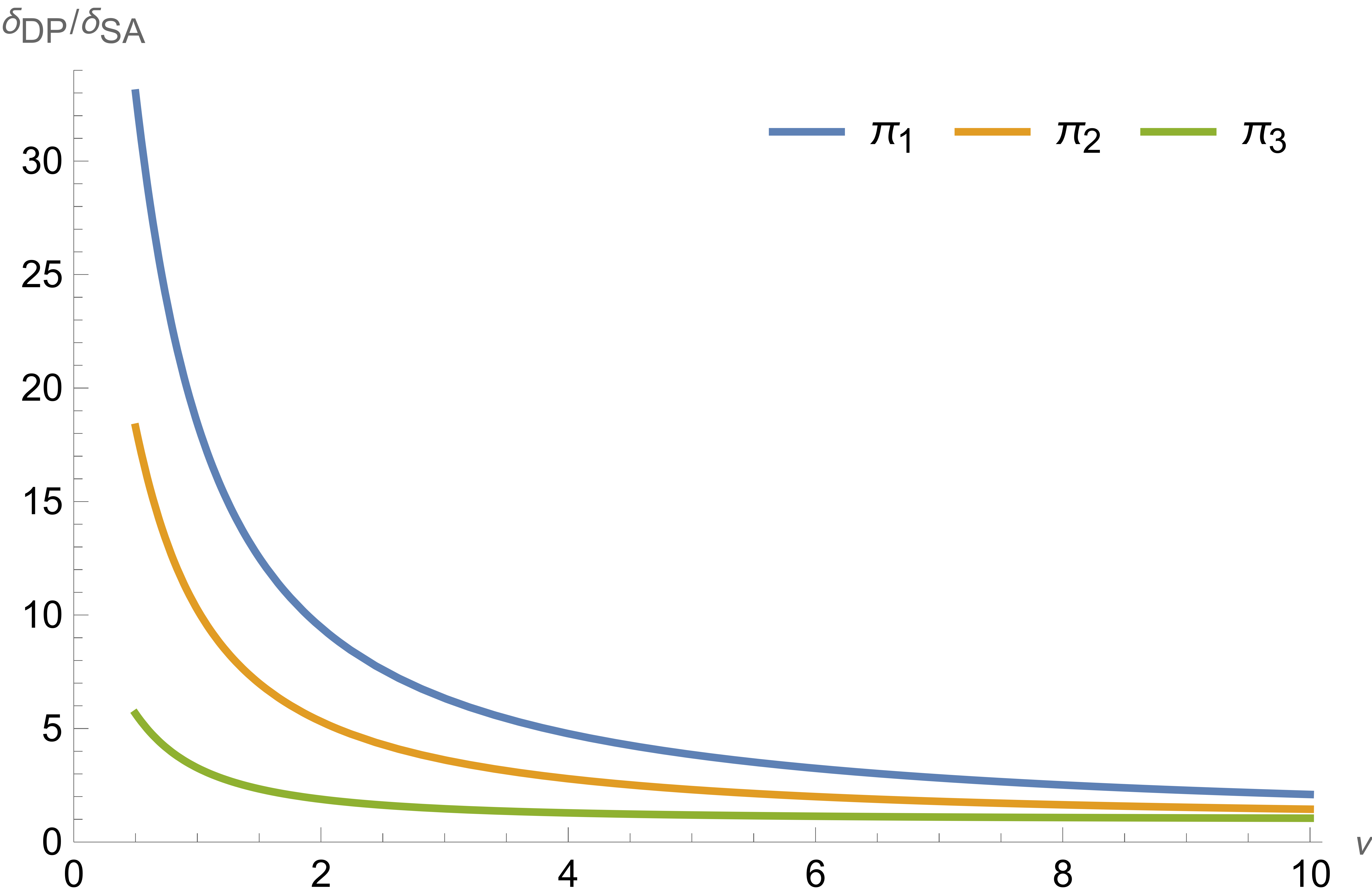}
        \caption{The ratio of the $\delta$ parameters for DP versus SP for different $\pi$-values.}
        \label{figure:gaus_dp:sa}
    \end{minipage}
\end{figure}
For Gaussian noise, see Fig.~\ref{figure:gaus_dp_vs_sa}: with
$\nu_1$,\KURZ{ that means}{} external noise equal to the sensitivity,
differential privacy achieves ${\color{blue!60!black} \delta = 0.38}$
and ${\color{green!60!black} \delta = 0.128}$ for $\nu_2$. 
\KURZ{In the}{For} statistical privacy\KURZ{ setting}{}, $\dd$ is much lower and mainly determined
by the internal noise if the external noise is moderate.
This holds unless $\pi$ gets close to $0$ when the internal noise vanishes.
%
Fig.~\ref{figure:gaus_dp:sa} shows the ratio between the $\dd$-values obtained for
these two privacy settings for different values of $\pi$, namely
${\color{blue!60!black} \pi_1 = 0.5}$,
${\color{yellow!30!orange} \pi_2 = 0.1}$ and ${\color{green!60!black} \pi_3 =0.01}$. 
The $x$-achses gives different strength of external noise 
where the factor $\nu$ ranges from $1$ to $10$.
Even for the maximum $10$ the ratio is significantly larger than $1$.
Note that  $\pi_3$ models a real rare event, where only $1 \%$ of all entries are
expected to be positive, this means $10$ entries out of $1000$.
Adding external noise with $\nu>>1$ degrades the utility of the answer to
almost useless.
\begin{figure}[htb]
    \centering
    \begin{minipage}{.45\textwidth}
         \centering
         \includegraphics[width=.95\linewidth]{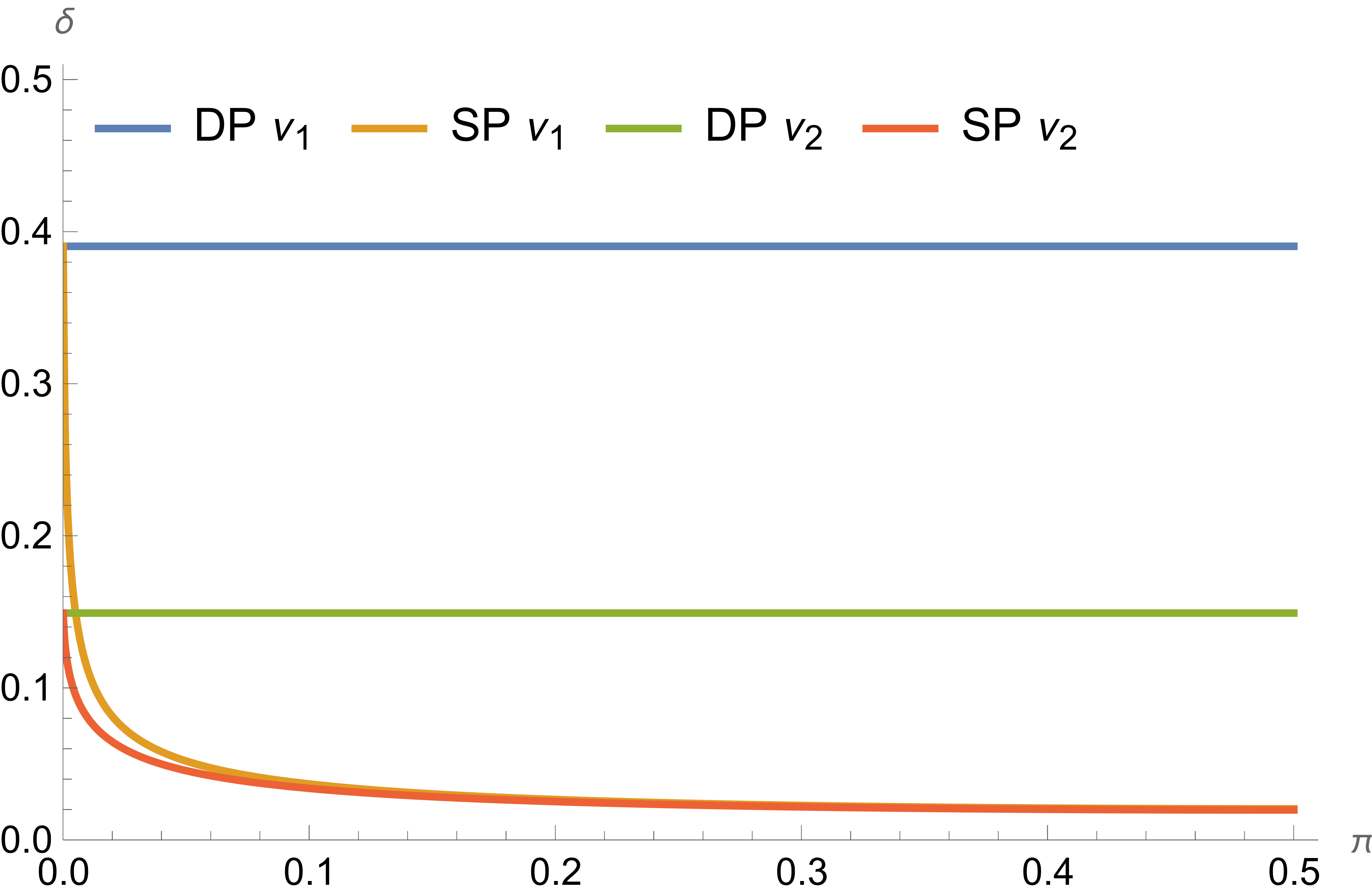}
         \caption{$\delta$ parameter with respect to $\pi$  for noise levels $\nu_1=1$ and  $\nu_2=3$.}
         \label{figure:lap_dp_vs_sa}
    \end{minipage}
    \hspace*{.05\textwidth}
    \begin{minipage}{.45\textwidth}
         \centering
         \includegraphics[width=.95\linewidth]{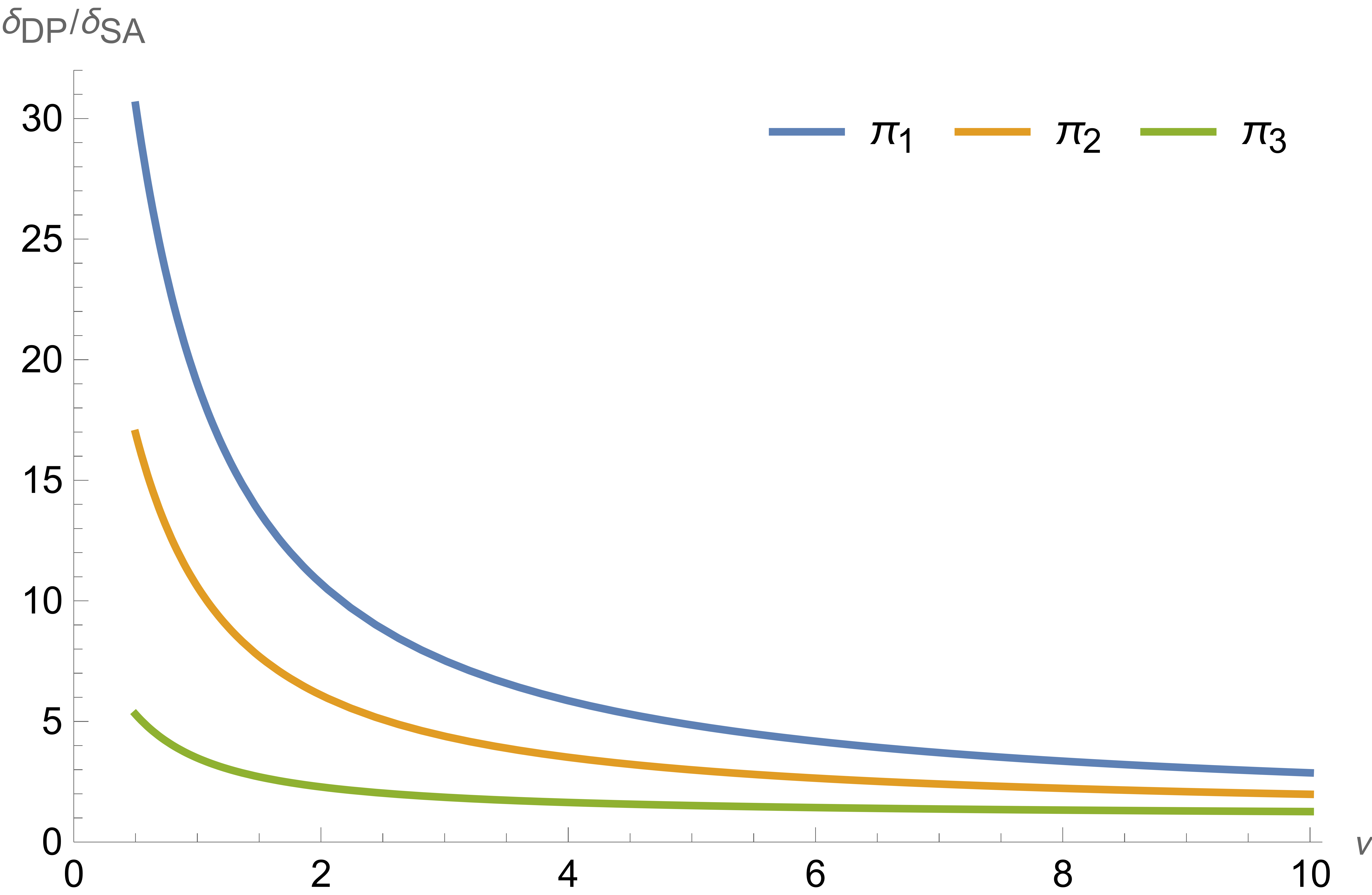}
         \caption{The ratio of the $\delta$ parameters for DP versus SP for different $\pi$-values.}
         \label{figure:lap_dp:sa}
    \end{minipage}
 \end{figure}
\KURZ{For Laplace noise our computations have used the same parameters.
Fig.~\ref{figure:lap_dp_vs_sa} and \ref{figure:lap_dp:sa} show that the results are quite similar.
Laplace noise seems to be slightly better suited in the statistical setting.}{For Laplace noise our computations have used the same parameters.
Fig.~\ref{figure:lap_dp_vs_sa} and \ref{figure:lap_dp:sa} show similar but
slightly better results for Lapalace noise.}
Even for quite small databases of size $n=100$ with much less entropy
the internal noise is quite helpful to guarantee privacy as Fig.~\ref{figure:lapgauss_small} illustrates.
The loss of privacy compared to $n=1000$ is quite small.
These results show that internal noise generated by the database distribution
has a strong effect on the  privacy guarantee obtained even in case of simple property queries.

\section{External Noise versus Subsampling for Statistical Privacy}
\label{SectionSubsam}

For differential privacy  the privacy amplification by subsampling 
has been studied in \cite{BBG20}, for noiseless privacy in \cite{BBG11}.
For differential privacy it has been shown
that subsampling  without replacement with rate
$\ll$ turns a $(\ee, \delta)$-differential privacy mechanism into one with parameters 
\begin{math}   (\klog(1+\lambda (e^\ee -1)) , \;  \lambda \; \delta ) \end{math}. 
Since for small $\ee$ the first expression can be approximated by 
\begin{math}
   \klog(1+\lambda (e^\ee -1)) \approx \ll \; \ee
\end{math} this yields a reduction
of the privacy parameters by the sampling rate.
Small $\ll$, however, yield high utility loss.

Can  the same reduction be achieved in the statistical setting?
Here one has to consider that smaller databases provide less entropy. 
Thus, when the subsampling mechanism generates a database of size $m<<n$
then the base should be the statistical privacy of databases of size $m$.
If this were much less than those for size $n$ then this reduces the effect of
subsampling. But the previous section has shown that the size reduction 
does not have a larger effect.

\KURZ{In case of property queries 
subsampling without replacement generates a hypergeometric distribution that 
for small sample size can be
approximated by a binomial distribution as we have already discussed above.
Since a binomial distribution converges to a normal distribution
the question arises what is the difference between subsampling and
Gaussian noise that generate the same standard deviation of the correct  value.

Furthermore, how does a combination of both behave?
In the limit as the convolution of two normal distributions one can expect 
normal distributed distortion with variance being the sum of both variances.
But for a precise answer again the same problem arises as for binomial distributions.
No exact formula seems to be known for the convolution of a
hypergeometric and a normal distribution.}{}

To generate real data on this issue we have done the following computations.
For a property query $\Fpi$ and  $\varepsilon = 0.01$ the corresponding $\delta$
is computed given a distortion by Laplace or Gaussian noise, resp.~subsampling with $\ll = 0.1$ 
such that the utility loss is identical for all three cases.
The loss is determined by the subsampling rate $\ll$\KURZ{. 
Then the strength of}{ and} the noise is adjusted\KURZ{ accordingly}{} to yield the same loss

\begin{figure}[ht!]
    \centering
    \begin{minipage}{.45\textwidth}
         \centering
         \includegraphics[width=.95\linewidth]{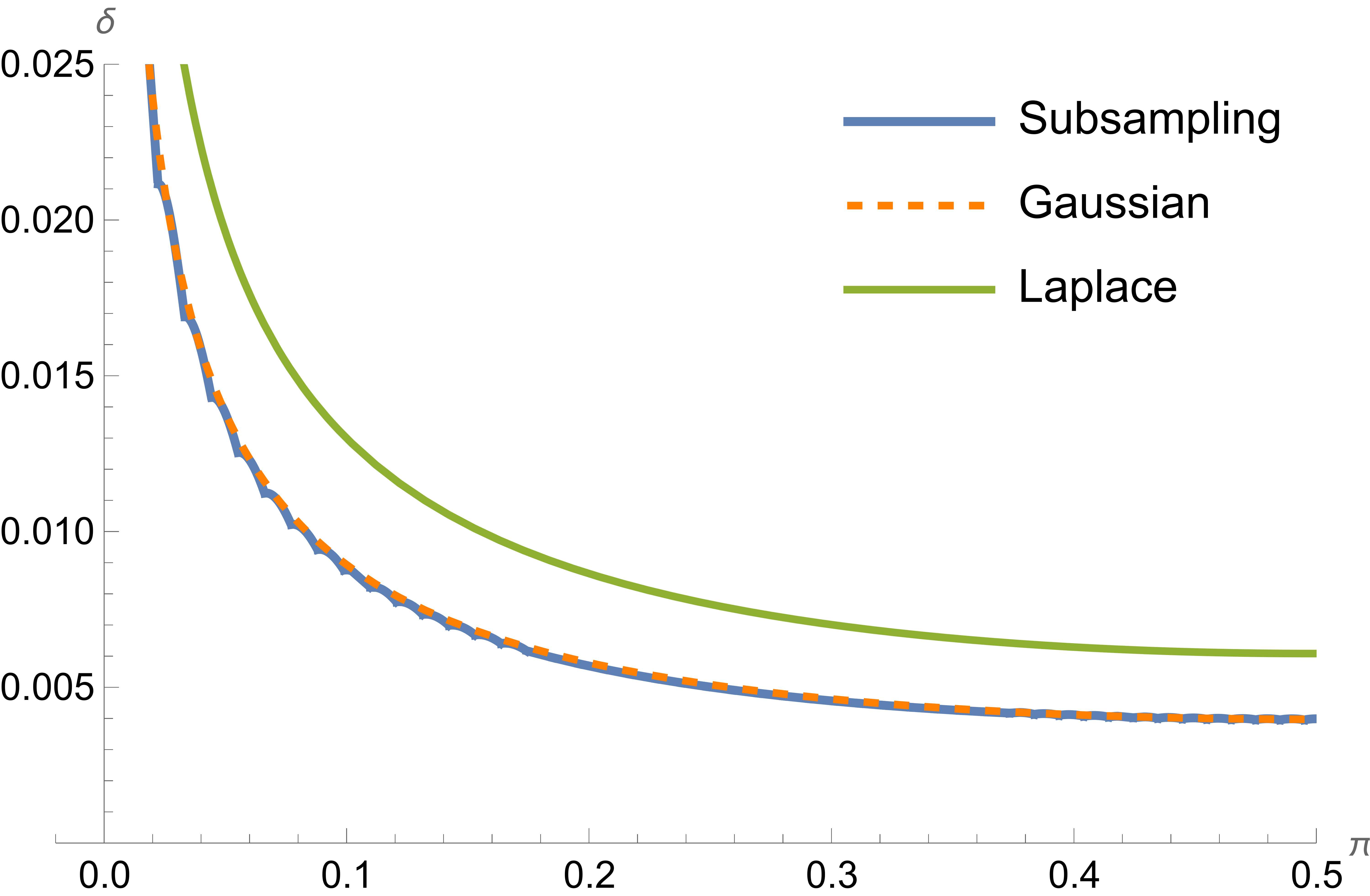} \\
         \caption{Comparison of the three mechanisms with respect to \KURZ{the property probability $\pi$}{$\pi$}}
         \label{figure:comparison_utility_loss_all_delta}
    \end{minipage}
    \hspace*{.05\textwidth}
    \begin{minipage}{.45\textwidth}
         \centering
         \includegraphics[width=.95\linewidth]{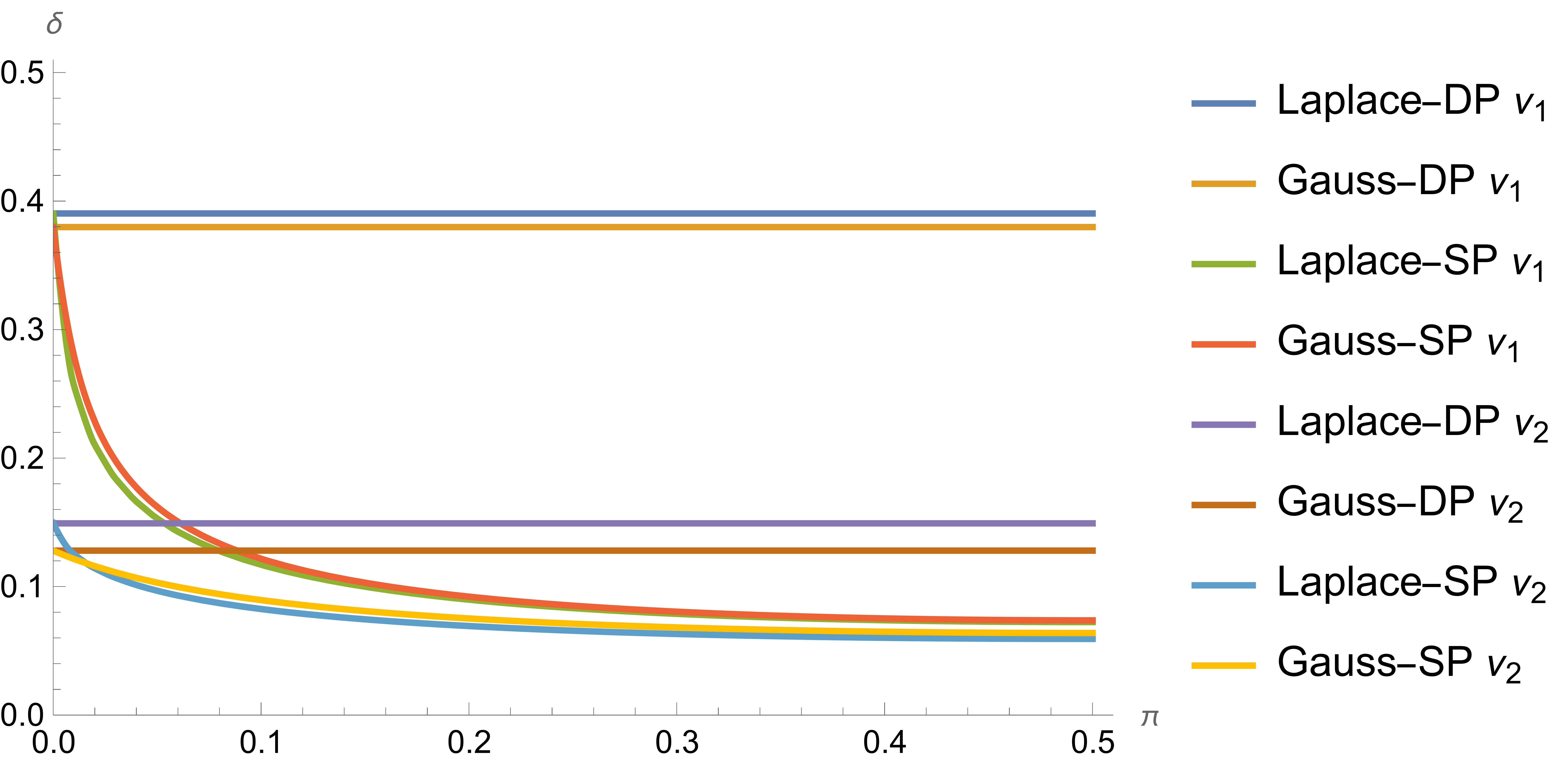}
         \caption{The ratio of the $\delta$ parameters for DP versus SP für $n=100$.}
         \label{figure:lapgauss_small}
    \end{minipage}
 \end{figure}


Fig.~\ref{figure:comparison_utility_loss_all_delta}
gives the results with respect to the property probability $\pi$.
It shows that Gaussian noise and subsampling practically achieve the same $\dd$-values 
even in this bounded scenario while Laplace noise is about  20\% worse.


Fig.~\ref{figure:comparison_variance_free_lambda} plots the dependency on $\ll$
for these three distortion techniques for the $\pi$-values $\pi_1 = 0.5$ and
$ \pi_2 = 0.1$. As above  the parameters are chosen in such a way that the utility loss
is identical in the three cases. Again the curves of subsampling and Gaussian noise fall together
while Laplace noise is always worse. The relative performance of Laplace noise deteriorates with 
smaller $\ll$.

\section{Conclusion}
In this report the privacy risk when accessing a database by  property queries
has been analyzed for adversaries  that know the distribution by which the database is generated.
We have named this scenario  statistical privacy
and  shown that good privacy bounds can already be achieved for moderate size databases.
They depend on the probability of the property\KURZ{ of interest}{}, but are useful unless the
probability gets so small that the expected number of positive entries is close
to 1.
\begin{figure}[ht] 
   \centering
   \begin{minipage}{.45\textwidth}
        \centering
        \includegraphics[width=.95\linewidth]{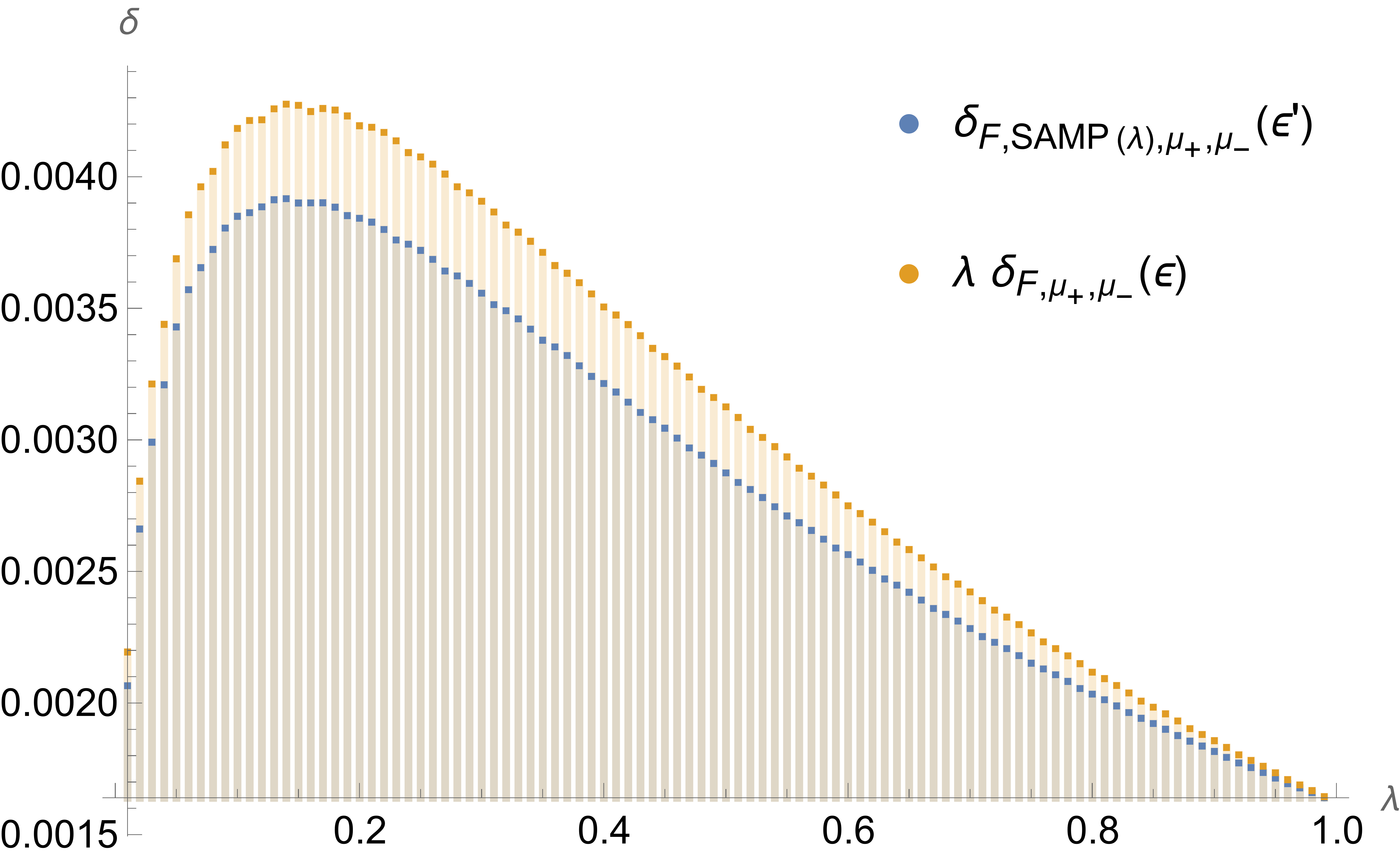}
        \KURZ{\caption{Results for $n = 1000$, $\pi = 0.5$ and $\ee = 0.1$
          the blue curve gives the $\dd$-values for subsampling a database of size $n$
          for $\ee'=  \log(1+\ll (e^\ee -1))$; note that $\dd$ increases  when decreasing $\ll$ 
          up to a certain point because  $\ee'$ gets smaller,
          the yellow curve shows the values for pure statistical privacy for databases of size $m$ multiplied by $\ll$}}{}
        \label{figure:subsampling_comparison}
   \end{minipage}
   \hspace*{.05\textwidth}
   \begin{minipage}{.45\textwidth}
        \centering
        \includegraphics[width=.95\linewidth]{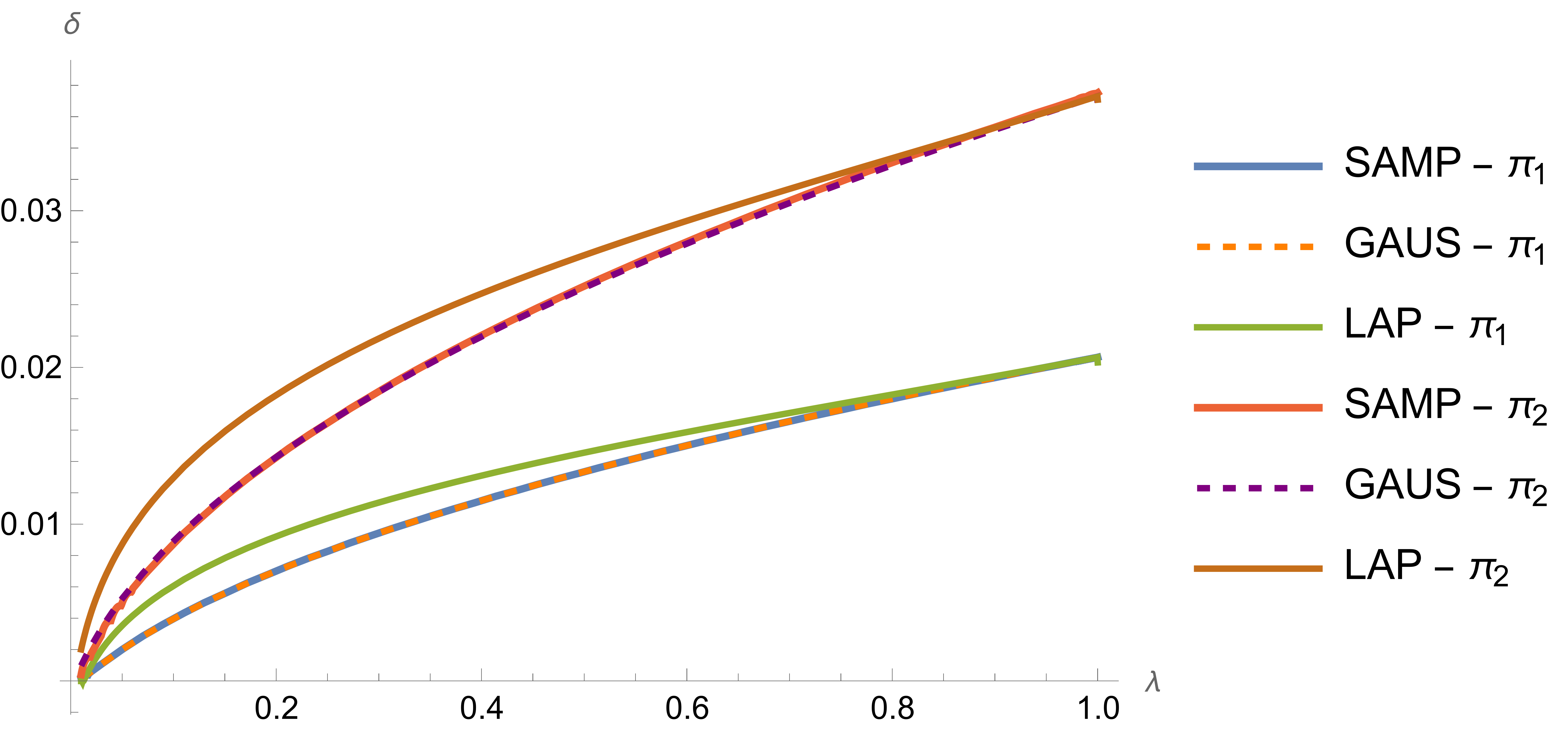}
        \KURZ{\caption{Comparison with respect to the subsampling rate $\ll$ for $\pi_1 = 0.5$ and $ \pi_2 = 0.1$}}{}
        \label{figure:comparison_variance_free_lambda}
   \end{minipage}
   \KURZ{}{\caption{Left: Results for $n = 1000$, $\pi = 0.5$ and $\ee = 0.1$
   the blue curve gives the $\dd$-values for subsampling a database of size $n$
   for $\ee'=  \log(1+\ll (e^\ee -1))$; note that $\dd$ increases  when decreasing $\ll$ 
   up to a certain point because  $\ee'$ gets smaller,
   the yellow curve shows the values for pure statistical privacy for
   databases of size $m$ multiplied by $\ll$. Right: Comparison with respect to the subsampling rate $\ll$ for $\pi_1 = 0.5$ and $ \pi_2 = 0.1$.}}
\end{figure}

For analyzing the addition of noise one has to determine the convolution
of the noise distribution with the distribution of the query under the database distribution.
This is an open mathematical problem for most pairs of distributions.


For distributional privacy there does not exist a general agreement for the choice of $\ee$.
Some real applications  have used values on the order of 10 or even larger.
This\KURZ{ does not make sense and}{} would yield a huge privacy leakage in this
worst case scenario. When this is\KURZ{ still}{} applied to large databases
it can only be justified by the condition that an adversary actually does
not know almost all entries -- thus we are in a statistical privacy setting.
The worst case scenario is not appropriate for many applications.



Privacy amplification by subsampling has been well studied 
for differential privacy.
In the statistical setting one has to analyze a mixture distribution
which can be quite complex in general.
Here we have restricted the query to a binomial distribution.
%
%
\KURZ{As already discussed s}{S}ubsampling with rate $\ll$ turns a $(\ee,
\delta)$-differential privacy mechanism into one where the parameters are of the form
\begin{math}\left(\log(1+\lambda (e^\ee -1)) ,\lambda \; \delta \right) \end{math},
which approximately for small $\ee$ can be simplified to $(\ll \; \ee, \:  \ll \; \dd)$.
And this bound is tight.

Does a comparable result hold for statistical privacy?
This is not obvious since  database of size $m<n$ have less entropy.
For property queries one has to compare 
${ \color{blue!60!black} \dd_{F, \SAMP(\ll),  \mu_+,\mu_-}(\ee')}$ 
for databases of size $n$ with the pure statistical privacy
${\color{yellow!50!orange}  \dd_{F,\mu_+,\mu_-}(\ee)}$ scaled by the rate $\ll$.
Fig.~\ref{figure:subsampling_comparison} shows the result of a specific test.
Here, subsampling gives even better results.
We conjecture that for property queries subsampling never leads to larger 
$\dd$-values and that this even holds more general, formally:



\KURZ{
\emph{
Let $\nu$ be a probability distribution of entries and 
$\nu^n, \; \nu^m$ be the product distribution of $\nu$ for databases of size $n$, resp.~$m$.
If $F$ is a symmetric query 
that is  $(\ee, \delta)$--pure statistical private 
for databases of size $m$ with distribution $\nu^m$
then   $(F, M_{\SAMP})$ with selection property $\ll = n/m$ is 
$(\log(1+\lambda (e^\ee -1)) ,\lambda  \; \dd )$--statistical private 
for databases of size $n$ with distribution $\nu^n$.
}
}{}

Further questions include other subsampling types, for example with replacement.
%
 We have shown  that for additive noise the utility loss is  equal to
the variance of the noise in the statistical setting.
Given a maximal acceptable utility loss Laplace noise seems to perform worse than
Gaussian noise in general.
This should be examined in more detail.

\newpage
%
%
%
%
%
\bibliographystyle{splncs04}
\bibliography{subsam}
\end{document}